\documentclass{article}
\usepackage[a4paper, total={7in, 10in}]{geometry}
\usepackage{graphicx} 
\usepackage{amsmath}
\usepackage{multirow}
\usepackage[authoryear]{natbib}
\usepackage{amssymb}
\usepackage{xcolor}
\usepackage{hyperref}
\usepackage{epstopdf} 
\usepackage{epsfig}

\newtheorem{theorem}{Theorem}

\newtheorem{corollary}[theorem]{Corollary}

\newtheorem{lemma}{Lemma}

\newtheorem{proposition}{Proposition}

\newenvironment{proof}[1][Proof]{\noindent\textbf{#1.} }{\ \rule{0.5em}{0.5em}}

\usepackage{etoolbox}
\apptocmd{\lim}{\limits}{}{}

\title{The green transition of firms: The role of evolutionary competition, adjustment costs, transition risk, and green technology progress\thanks{The paper benefited from the comments of the participants at the 
27th Annual Workshop on Economics with Heterogeneous Interacting Agents (WEHIA 2024) in Bamberg. We thank Tomasz Makarewicz and Kerstin H\"{o}tte for their valuable insights and stimulating discussions.}}
\author{Davide Radi$^{\dagger}$\footnote{Corresponding author: Davide Radi, Catholic University of Sacred Heart, Department of Mathematics for Economic, Financial and Actuarial Sciences (DiMSEFA), Via Necchi 9, 20123 Milan, Italy. Email: davide.radi@unicatt.it.} and Frank Westerhoff$^{+,}$}
\date{%
\footnotesize
$^{\dagger}$Department of Mathematics for Economic, Financial and Actuarial Sciences (DiMSEFA), Catholic University of Sacred Heart, Milan, Italy. \\
$^{+}$Department of Economics, University of Bamberg, Bamberg, Germany.\\
October 2024}

\begin{document}

\maketitle

\begin{abstract}
We propose an evolutionary competition model to investigate the green transition of firms, highlighting the role of adjustment costs, dynamically adjusted transition risk, and green technology progress in this process. Firms base their decisions to adopt either green or brown technologies on relative performance. To incorporate the costs of switching to another technology into their decision-making process, we generalize the classical exponential replicator dynamics. Our global analysis reveals that increasing transition risk, e.g., by threatening to impose stricter environmental regulations, effectively incentivizes the green transition. Economic policy recommendations derived from our model further suggest maintaining high transition risk regardless of the industry’s level of greenness. Subsidizing the costs of adopting green technologies can reduce the risk of a failed green transition. While advances in green technologies can amplify the effects of green policies, they do not completely eliminate the possibility of a failed green transition. Finally, evolutionary pressures favor the green transition when green technologies are profitable.

\medskip


\medskip


\end{abstract}

\section{Introduction}

In line with the European Green Deal and the 2015 Paris Climate Agreement, the European Union (EU) aims to achieve climate neutrality by 2050, striving for a net-zero greenhouse gas emissions economy. The transition to a low-carbon economy requires firms to adopt green technologies (\cite{Rodrik2014}), which may reduce their profitability (\cite{MorgensternPizerShih2001}).\footnote{For surveys on the impact of environmental regulations, see \citet{JaffePetersonPortneyStavins1995} and \citet{Schmalensee1994}} Unfortunately, this financial burden may slow or even prevent the green transition.

In this paper, we propose a novel model of evolutionary competition to enhance our understanding of the green transition. In addition to evolutionary competition, our model incorporates three key factors: adjustment costs, dynamically adjusted transition risk, and green technology progress. Adjustment costs represent the one-time fixed costs that firms incur when they change their production technology. Transition risk refers to the possibility that new green regulations will make activities that are incompatible with a low-carbon economy more expensive, e.g., by imposing a carbon tax, or force brown firms to change production technologies and bear the associated costs (\cite{DunzNaqviMonasterolo2021}, \cite{BoltonKacperczyk2023}). This risk is dynamically adjusted when its intensity depends on the level of greenness of the industry. In contrast, green technology progress is responsible for increasing the returns to adopting green production technologies (\cite{PorterLinde1995}, \cite{Zeppini2015}).

Our goal is to explore the role of these factors in promoting the adoption of green technologies in the context of evolutionary competition. Evolutionary strategy selection is a characteristic of manager-led firms (\cite{Hirshleifer1993}) and is frequently employed in studies of environmental issues (\cite{Zeppini2015}, \cite{AntociBorghesiIannucciSodini2022}, \cite{AntociBorghesiGaldiVergalli2022}, \cite{IannucciTampieri2024}). The dynamic framework we employ is a generalized version of the exponential replicator dynamics (\cite{HofbauerWeibull1996}, \cite{HofbauerSigmund2003}), which introduces an essential innovation to account for the costs associated with changing production technologies.

Our investigation follows a three-stage approach. In the first stage, we consider a simplified model with no adjustment costs, dynamically adjusted transition risk, and green technology progress. The firm’s decision-making process follows a classical one-dimensional discrete-time exponential replicator dynamics. In the second stage, we introduce adjustment costs while still excluding dynamically adjusted transition risk and green technology progress. The firm's decision-making process is now driven by a generalized version of the classical one-dimensional discrete-time exponential replicator dynamics, introduced here for the first time. This approach assigns adjustment costs to firms that change technologies and incorporates these costs into their decision-making process. In the third stage, we analyze the full model, which includes adjustment costs, dynamically adjusted transition risk, and green technology progress. This model is represented by a two-dimensional discrete-time exponential replicator dynamics, also introduced here for the first time.

Our main findings are as follows. When adjustment costs, dynamically adjusted transition risk, and green technology progress are excluded, a green transition occurs either when the green technology is the most profitable or when green regulations make it so.

Adding adjustment costs to the model introduces uncertainty about which technology will be adopted; the higher the cost of switching, the lower the likelihood of adoption. As a result, a green transition may fail to materialize even when the green technology is the most profitable. To mitigate this risk, regulations should impose additional costs on brown production to ensure that green technology yields the highest net profit. Interestingly, the propensity of firms to switch to the most profitable technology increases the likelihood of a green transition when the profit gap of remaining a green firm is larger than that of remaining a brown firm, and decreases the likelihood otherwise. From a policy perspective, the greater the propensity of firms to adopt the most profitable technology, i.e., the greater the evolutionary pressure, the less extensive green economic policies need to be to ensure a successful green transition.

Taking into account dynamically adjusted transition risk and green technology progress, our model offers several scenarios. We employ analytical arguments to characterize all of them and identify the economic conditions that facilitate a green transition. A key finding is that dynamically adjusted transition risk does not preclude a green transition, and green technology progress does not eliminate the possibility of transition failure. Furthermore, even when the green technology is the most profitable, a green transition may not occur, due to adjustment costs. Nevertheless, the results confirm that an increasing transition risk lowers the likelihood of a failed green transition. The economic policy recommendation is to maintain a high level of transition risk regardless of the number of brown firms in the industry. This is in line with \citet{vandenBergh2012}, who emphasizes that mitigating climate change is a public good and requires effective, systemic policies to promote a low-carbon transition and avoid ‘escape routes’ associated with partial solutions. As highlighted in \citet{AlessiBattiston2022}, keeping transition risk high regardless of the level of greenness is a conservative approach that also reduces the risk of green-washing. 

In conclusion, the only green economic policy that guarantees a successful green transition is the threat of introducing stronger green policies, such as the imposition of a tax on brown production, which significantly increases transition risk. In contrast, subsidizing the green technology only ensures the green transition if it becomes more profitable than the brown technology.

The road map of the paper is as follows. Section \ref{LitRew} briefly reviews the related literature and empirical evidence. Section \ref{ModelSetup} introduces our model and two nested versions. Section \ref{NestedVersion1} investigates the dynamics of the simplest nested version of the model, which excludes adjustment costs, dynamically adjusted transition risk, and green technology progress. Section \ref{NestedVersion2} examines the dynamics of the nested version that includes adjustment costs but excludes dynamically adjusted transition risk and green technology progress. The aim is to highlight the role of adjustment costs between green and brown production. Section \ref{fullfledgedmodel} explores the dynamics of the full model, focusing on the role of adjustment costs, dynamically adjusted transition risk, and green technology progress. The implications of possible green economic policies are also discussed. Section \ref{Conc} concludes. Appendix \ref{AppA} contains analytical results for the full model and numerical tests that confirm the robustness of the results. All technical proofs are provided in Appendix \ref{AppB}.

\section{Related literature and empirical evidence}\label{LitRew}

\subsection{Adjustment costs, dynamically adjusted transition risk, and green technology progress}

Adjustment costs are a key factor in lost profits during the green transition. The adoption of green technologies or changing production technologies in general, involves significant upfront fixed costs. In addition, productivity losses are another effect of environmental regulations\footnote{According to the \citet{OTAITE586}, evidence from productivity growth in the 1970s and 1980s suggests that between 8\% and 44\% of industry-level declines in productivity growth were attributable to environmental regulation.}, along with various direct and indirect production costs, caused by factors such as crowding out other productive investments (\cite{Rose1983}), discouraging investment in more efficient facilities (\cite{Gruenspecht1982}, \cite{NelsonTietenbergDonihue1993}), and meeting pollution control requirements (\cite{JoshiLaveShihMcMichael1997}).

On the other hand, firms operating with brown technology are exposed to the risk of more stringent environmental regulations, commonly referred to as climate policy transition risk (\cite{FriedNovanPeterman2022}). This risk negatively affects the profit expectations of brown firms. Climate policy transition risk encompasses the possibility of new environmental regulations that require brown firms to adopt green technologies and bear the associated adjustment costs. It also includes the risk of carbon pricing and stringent green economic policies designed to internalize carbon externalities, such as a carbon tax or tradable emissions permits, see, e.g., \citet{Bowen2011}, \citet{AntociBorghesiIannucciSodini2021}, and \citet{IannucciTampieri2024}.

Importantly, climate policy transition risk may be state-dependent, meaning that public attention, and hence political pressure for a green transition, declines as the industry becomes greener. This reduction in environmental concern has a positive impact on the market value of carbon-intensive firms (\cite{EngleGiglioKellyLeeStroebel2020}) and reduces the positive performance of green stocks, negative greenium, (\cite{PastorStambaughTaylor2022}). Our model captures this dynamically adjusted transition risk by assuming that as fewer firms rely on brown technologies, their production becomes increasingly profitable.

Conversely, clean technologies tend to have steeper learning curves than incumbent dirty technologies, see \citet{McNerneyFarmerRedneraTrancik2011}, and the cost of green technologies decreases as their adoption grows.\footnote{\citet{Allen2009} emphasizes that during the Industrial Revolution, having a large enough market played a key role in rewarding developers for perfecting technology. Specifically, \citet{Allen2009} argues that it is the demand for new techniques that sustains the acceleration of technological progress, making it more profitable to adopt. \citet{PearsonFoxon2012} stress that this is a relevant aspect to avoid a brown technology lock-in, see \citet{Arthur1989}, and to promote a low-carbon industrial revolution.} This reduction results from advances in green technologies, including economies of scale in the production of low-carbon technologies (\cite{PearsonFoxon2012}), technical improvements (\cite{Arthur1989}, \cite{Unruh2000}, \cite{Zeppini2015}), and positive externalities or spillover effects as production scales up and firms accumulate experience (\cite{KatzShapiro1985}, \cite{GrubbKohlerAnderson2002}).

We model this aspect by assuming that the profitability of green technologies increases when all firms adopt them, and we test the hypothesis that green technology progress favors a green transition (\cite{PearsonFoxon2012}). Finally, the green transition can be profitable for firms because it can create new opportunities and capture new market share (\cite{Porter1991}, \cite{PorterLinde1995}). Our model is flexible enough to capture this scenario, as well as to highlight how advances in green technologies interact with green economic policies, amplifying their impact.\footnote{Supporting low-carbon technologies by increasing the risk of climate policies helps achieve sufficient market penetration and experience to compete with incumbent technologies that have long enjoyed the benefits of technological and institutional returns to scale \citet{PearsonFoxon2012}.}

\subsection{Policy implications}

Our results build on and extend the work of \citet{Zeppini2015} by introducing a more comprehensive framework that incorporates adjustment costs and dynamically-adjusted transition risk. This adjustment requires alternative modeling choices in order to maintain analytical tractability. Specifically, these differences related to the firm’s decision-making process, which in \citet{Zeppini2015} follows discrete choice dynamics, a micro-founded model of social learning that has been widely adopted in the economic and financial literature, see \citet{BrockHommes1997,BrockHommes1998}. In contrast, our model is characterized by evolutionary dynamics. Instead of incorporating adjustment costs and dynamically adjusted transition risk, \citeauthor{Zeppini2015}’s (\citeyear{Zeppini2015}) model accounts for endogenous technological progress, environmental policy, and social learning. The results are also similar, with a carbon lock-in equilibrium that requires a carbon tax to overcome and enforce the green transition.

Moreover, our findings on the role of transition risk are consistent with those of \citet{FriedNovanPeterman2022}, who find that such risks reduce carbon emissions by encouraging investment in cleaner technologies. Interestingly, \citet{FriedNovanPeterman2022} emphasize that the risk aversion of firms causes the green transition to accelerate more in response to transition risk than in response to a small, certain carbon tax. Their findings are based on a dynamic general equilibrium model that incorporates beliefs about the likelihood that the government will adopt a climate policy that transitions the economy to a lower-carbon steady state. Our setup, on the other hand, belongs to the realm of partial equilibrium models and is based on evolutionary competition and firms seeking the best relative performance. The fact that different modeling approaches yield the same result confirms its robustness. Anecdotal evidence also supports this, as the announcement and subsequent implementation of a carbon tax has led to a significant decline in the value of fossil capital, see, e.g., \citet{vanderPloegRezai2020}.

\citet{CarattiniHeutelMelkadze2023} point out that transition risk, such as the imposition of a carbon tax, is often combined with macro-prudential policies (e.g., a tax on brown assets for banks) to mitigate the risk of macroeconomic instability. Thus, a brown firm must consider not only the risk of a carbon tax, but also the risk that its carbon-intensive assets may become stranded, i.e., lose most of their economic value. It is also worth noting that these results contradict the predictions of the green paradox literature (e.g. \cite{Sinn2008}), which argues that the risk of future climate policies would increase current emissions by strengthening incentives to extract fossil fuels, thereby expanding supply. These conflicting results may be due to specific conditions that transition risk must satisfy in order to meaningfully affect investments (\cite{FriedNovanPeterman2022}): 1) the likelihood that a climate policy will be adopted in the near future cannot be trivially small, and 2) firms must believe that the climate policy, if implemented, will be stringent enough to change the return on investment. These conditions are confirmed by our analysis and are typically met in practice, as evidenced by surveys showing that firms anticipate business risks from existing or expected carbon regulations by shifting investments to projects that would be competitive in a carbon-constrained future (\cite{Ahluwalia2017}).

Our analysis also echoes the findings of \citet{EngleGiglioKellyLeeStroebel2020}, who show that the stock prices of firms most exposed to transition risk perform relatively worse during periods when attention to regulatory risk is higher (e.g., after negative news about climate change) or after events that are likely to increase the perceived likelihood of future climate regulation. Our contribution is also related to the literature on partial equilibrium models that explore the impacts of environmental policy uncertainty on investment and location decisions (\cite{Xepapadeas2001}, \cite{PommeretSchubert2018}). In contrast to policy-driven uncertainty, a much larger body of literature focuses on how optimal environmental policies are affected by uncertainty arising from often irreversible environmental shocks (\cite{LemoineTraeger2014}).

\section{Model Setup}\label{ModelSetup}

Consider a setup where two firms (players) can choose to adopt either a green or a brown technology.\footnote{According to the EU Taxonomy for sustainable economic activities, see, e.g., \citet{AlessiBattiston2022}, a green technology can be considered such if it contributes to at least one of the following environmental objectives: (i) climate change mitigation; (ii) climate change adaptation; (iii) sustainable use and protection of water and marine resources; (iv) transition to a circular economy; (v) pollution prevention and control; and (vi) protection and restoration of biodiversity and ecosystems.} A firm that adopts a green technology is denoted by $G$, a firm that adopts a brown technology is denoted by $B$. A green technology ensures a payoff $\Pi\left(G,X\right)$ while a brown technology ensures a payoff of $\Pi\left(B,X\right)$. These payoffs depend on $X$, which is the technology adopted by the competitor.

The dependence of the profit function on the technology adopted by the competitor is due to two factors: green technology progress and a dynamically adjusted transition risk. Firms adopt green technologies produced by another industry. Factors such as experience, spillover effects, accumulated knowledge, and economies of scale contribute to green technology progress, making the green technology more profitable when both firms adopt it, see, e.g., \citet{Unruh2000}, \citet{Allen2009}, \citet{PearsonFoxon2012} and \citet{Zeppini2015}. Then we have that:
\begin{equation}
\Pi\left(G,G\right)=\Pi^{GG}>\Pi\left(G,B\right) =\Pi^{GB} 
\end{equation}
Climate policy transition risk is the risk that a new stringent environmental regulation is introduced and reduces the profitability of the brown technology, see, e.g., \citet{FriedNovanPeterman2022} and \citet{CarattiniHeutelMelkadze2023}. We model this type of risk by assuming that it is higher when both firms adopt a brown technology, which we call dynamically adjusted transition risk.\footnote{A dynamically adjusted transition risk is the result of adaptive climate policies aimed at facilitating a smooth green transition.} This modeling assumption echoes the findings of \citet{EngleGiglioKellyLeeStroebel2020}, which show that the stock prices of firms most exposed to transition risk perform relatively worse during periods when attention to regulatory risk is plausibly higher (e.g., when there is negative news about climate change) or after events that are likely to increase the perceived likelihood of future climate regulation. Therefore:
\begin{equation}
\Pi\left(B,G\right)=\Pi^{BG}>\Pi\left(B,B\right)=\Pi^{BB}  
\end{equation}
There are no restrictions on the profit gap between green and brown productions. Evidence and surveye suggest that this gap is negative, see, e.g., \citet{Schmalensee1994}, \citet{JaffePetersonPortneyStavins1995}, \citet{MorgensternPizerShih2001} and \citet{OTAITE586}. However, there are also arguments supporting the idea that environmental regulation could be costless or even profitable, see, e.g., \citet{Porter1991} and \citet{PorterLinde1995}. The underlying notion is that environmental requirements can encourage plant managers to innovate and thus offset the costs associated with environmental protection. Finally, we assume that a firm has to pay an adjustment cost of $C^{B}$ to switch from a green to a brown technology, while it has to pay an adjustment cost of $C^{G}$ to switch from a brown to a green technology.

Inspired by the management literature that emphasizes how managers tend to make decisions in order to achieve the best relative performance, see, e.g., \citet{Hirshleifer1993}, we assume that firms update their decisions according to a two-population dynamic game, in particular, according to an exponential replicator dynamics adjusted to account for the costs of switching strategies. Let $P_{i}\left(G_{t}^{i}\right) = \eta_{t}^{i,G}$ be the probability that player $i$, with $i\in\left\{1,2\right\}$, adopts the green technology at time $t$ and let $P_{-i}\left(G_{t}^{-i}\right) =\eta_{t}^{-i,G}$ be the probability that the opponent adopts the green technology at time $t$. The probability that agent $i$ will adopt the green technology is updated at each time step $t$ according to the following revision protocol:
\begin{equation}
\eta^{i,G}_{t+1} = P_{i}\left(\left.G_{t+1}^{i}\right|G_{t}^{i};\eta^{-i,G}_{t}\right) \eta^{i,G}_{t} +  P_{i}\left(\left.G_{t+1}^{i}\right|B_{t}^{i};\eta^{-i,G}_{t}\right) \left(1-\eta^{i,G}_{t}\right)\text{,}
\end{equation}
where
\begin{equation}
P_{i}\left(\left.G_{t+1}^{i}\right|G_{t}^{i};\eta^{-i,G}_{t}\right) = \frac{\eta^{i,G}_{t} \exp\left(\beta E\left[\left.\Pi^{G}\right|\eta^{i,G}_{t}=1,\eta^{-i,G}_{t}\right]\right)}{\eta^{i,G}_{t}\exp\left(\beta E\left[\left.\Pi^{G}\right|\eta^{i,G}_{t}=1,\eta^{-i,G}_{t}\right]\right)+\left(1-\eta^{i,G}_{t}\right)\exp\left(\beta E\left[\left.\Pi^{B}\right|\eta^{i,G}_{t}=1,\eta^{-i,G}_{t}\right]\right)}
\end{equation}
and
\begin{equation}
P\left(\left.G_{t+1}^{i}\right|B_{t}^{i};\eta^{-i,G}_{t}\right) = \frac{\eta^{i,G}_{t}  \exp\left(\beta E\left[\left.\Pi^{G}\right|\eta^{i,G}_{t}=0,\eta^{-i,G}_{t}\right]\right)}{\eta^{i,G}_{t}  \exp\left(\beta E\left[\left.\Pi^{G}\right|\eta^{i,G}_{t}=0,\eta^{-i,G}_{t}\right]\right)+\left(1-\eta^{i,G}_{t} \right)\exp\left(\beta E\left[\left.\Pi^{B}\right|\eta^{i,G}_{t}=0,\eta^{-i,G}_{t}\right]\right)}\text{.}
\end{equation}
The intensity of the firm's choice is represented by $\beta$. The higher $\beta$ is, the higher the propensity to switch to the more profitable technology in the current period.
In accordance with the hypotheses made and consistent with evolutionary dynamics, we assume that player $i$'s expectations at time $t$ for profits at time $t+1$ are:
\begin{equation}
\begin{array}{lll}
E\left[\left.\Pi^{G}\right|\eta^{i,G}_{t}=1,\eta^{-i,G}_{t}\right] & =  & \Pi^{GG}\eta^{-i,G}_{t}+\Pi^{GB}\left(1-\eta^{-i,G}_{t}\right)  \\
\\
E\left[\left.\Pi^{G}\right|\eta^{i,G}_{t}=0,\eta^{-i,G}_{t}\right] & =  & \left(\Pi^{GG}-C^{G}\right)\eta^{-i,G}_{t}+\left(\Pi^{GB}-C^{G}\right)\left(1-\eta^{-i,G}_{t}\right) =  E\left[\left.\Pi^{G}\right|\eta^{i,G}_{t}=1,\eta^{-i,G}_{t}\right]-C^{G}\\
\\
E\left[\left.\Pi^{B}\right|\eta^{i,G}_{t}=0,\eta^{-i,G}_{t}\right] & =  & \Pi^{BG}\eta^{-i,G}_{t}+\Pi^{BB}\left(1-\eta^{-i,G}_{t}\right)  \\
\\
E\left[\left.\Pi^{B}\right|\eta^{i,G}_{t}=1,\eta^{-i,G}_{t}\right] & =  & \left(\Pi^{BG}-C^{B}\right)\eta^{-i,G}_{t}+\left(\Pi^{BB}-C^{B}\right)\left(1-\eta^{-i,G}_{t}\right) = E\left[\left.\Pi^{B}\right|\eta^{i,G}_{t}=0,\eta^{-i,G}_{t}\right]-C^{B}\\
\end{array}\text{.}
\end{equation}
Then, the dynamics of $\eta^{i,G}_{t}$, with $i =1,2$, is given by
\begin{equation}\label{MAINMAP}
\resizebox{1\hsize}{!}{$
\begin{array}{lll}
\eta^{1,G}_{t+1} & = &  \frac{\left(\eta^{1,G}_{t}\right)^2}{\eta^{1,G}_{t}+\left(1-\eta^{1,G}_{t}\right)\exp\left(\beta \left(\left(\Pi^{BG}-\Pi^{GG}\right)\eta^{2,G}_{t}+\left(\Pi^{BB}-\Pi^{GB}\right)\left(1-\eta^{2,G}_{t}\right)-C^{B}\right)\right)} +  \frac{\eta^{1,G}_{t}\left(1-\eta^{1,G}_{t}\right)}{\eta^{1,G}_{t}+\left(1-\eta^{1,G}_{t}\right)\exp\left(\beta\left(\left(\Pi^{BG} -\Pi^{GG}\right)\eta^{2,G}_{t}+\left(\Pi^{BB}-\Pi^{GB}\right)\left(1-\eta^{2,G}_{t}\right)+C^{G}\right)\right)} \\
\\
\eta^{2,G}_{t+1} & = & \frac{\left(\eta^{2,G}_{t}\right)^{2}}{\eta^{2,G}_{t}+\left(1-\eta^{2,G}_{t}\right)\exp\left(\beta \left(\left(\Pi^{BG}-\Pi^{GG}\right)\eta^{1,G}_{t}+\left(\Pi^{BB}-\Pi^{GB}\right)\left(1-\eta^{1,G}_{t}\right)-C^{B}\right)\right)} +  \frac{\eta^{2,G}_{t}\left(1-\eta^{2,G}_{t}\right)}{\eta^{2,G}_{t}+\left(1-\eta^{2,G}_{t}\right)\exp\left(\beta\left(\left(\Pi^{BG} -\Pi^{GG}\right)\eta^{1,G}_{t}+\left(\Pi^{BB}-\Pi^{GB}\right)\left(1-\eta^{1,G}_{t}\right)+C^{G}\right)\right)}
\end{array}
$}
\end{equation}
where all parameters are assumed to be positive, that is $\Pi^{BB},\Pi^{BB},\Pi^{GG},\Pi^{GB},C^{B},C^{G},\beta>0$. 

To describe the probabilities, we consider and investigate the dynamics of the model in the unit box $\mathcal{D}=\left[0,1\right]^{2}$. However, before focusing on the full version of our model, we consider two nested cases.

The first benchmark setup is the nested model with adjustment costs but without green technology progress and dynamically adjusted transition risk. The dynamics of this model is driven by a one-dimensional map and takes the form of an exponential replicator dynamics adjusted to account for the switching of strategies. We thus have that
\begin{equation}
\Pi\left(G,G\right)=\Pi^{GG}=\Pi\left(G,B\right) =\Pi^{GB} =\Pi^{G}
\end{equation}
and
\begin{equation}
\Pi\left(B,G\right)=\Pi^{BG}=\Pi\left(B,B\right)=\Pi^{BB} =\Pi^{B}\text{.}
\end{equation}
Then, player $i$'s expected profits do not depend on the opponent's choice:
\begin{equation}
\begin{array}{lll}
E\left[\left.\Pi^{G}\right|\eta^{i,G}_{t}=1,\eta^{-i,G}_{t}\right] & =  & \Pi^{GG}\eta^{-i,G}_{t}+\Pi^{GB}\left(1-\eta^{-i,G}_{t}\right) = \Pi^{G}\eta^{-i,G}_{t}+\Pi^{G}\left(1-\eta^{-i,G}_{t}\right) = \Pi^{G}  \\
\\
E\left[\left.\Pi^{G}\right|\eta^{i,G}_{t}=0,\eta^{-i,G}_{t}\right] & =  &   E\left[\left.\Pi^{G}\right|\eta^{i,G}_{t}=1,\eta^{-i,G}_{t}\right]-C^{G} =  \Pi^{G} - C^{G}\\
\\
E\left[\left.\Pi^{B}\right|\eta^{i,G}_{t}=0,\eta^{-i,G}_{t}\right] & =  & \Pi^{BG}\eta^{-i,G}_{t}+\Pi^{BB}\left(1-\eta^{-i,G}_{t}\right) = \Pi^{B}\eta^{-i,G}_{t}+\Pi^{B}\left(1-\eta^{-i,G}_{t}\right) = \Pi^{B}  \\
\\
E\left[\left.\Pi^{B}\right|\eta^{i,G}_{t}=1,\eta^{-i,G}_{t}\right] & =  &   E\left[\left.\Pi^{B}\right|\eta^{i,G}_{t}=0,\eta^{-i,G}_{t}\right]-C^{B} = \Pi^{B} - C^{B}\\
\end{array}
\end{equation}
Therefore, the dynamics of $\eta^{i,G}_{t}$ and $\eta^{-i,G}_{t}$ are independent. It follows that the evolutionary model \eqref{MAINMAP} is a decoupled two-dimensional dynamical system. Its dynamics is fully described by the following one-dimensional replicator dynamics, adjusted to account for the switching of strategies:
\begin{equation}\label{MAINMAPNested1}
\eta^{i,G}_{t+1}  =  \eta^{i,G}_{t}\frac{\eta^{i,G}_{t} }{\eta^{i,G}_{t}+\left(1-\eta^{i,G}_{t}\right)\exp\left(\beta \left(\Pi^{B}-\Pi^{G}-C^{B}\right)\right)} + \left(1-\eta^{i,G}_{t}\right) \frac{\eta^{i,G}_{t} }{\eta^{i,G}_{t}  +\left(1-\eta^{i,G}_{t}\right)\exp\left(\beta \left(\Pi^{B} - \Pi^{G} + C^{G}\right)\right)} 
\end{equation}

The second benchmark setup is the nested model without green technology progress, dynamically adjusted transition risk, and adjustment costs, i.e., $C^{B}=C^{G}=0$. Then, the probability of adopting a green technology in the next period does not depend on the need to switch strategies, and the one-dimensional model \eqref{MAINMAPNested1} reduces to the classical exponential replicator dynamics, as proposed by \citet{HofbauerSigmund2003}: 
\begin{equation}\label{MAINMAPNested2}
\eta^{i,G}_{t+1}  = \frac{\eta^{i,G}_{t}}{\eta^{i,G}_{t} +\left(1-\eta^{i,G}_{t}\right)\exp\left(\beta\left( \Pi^{B} - \Pi^{G}\right)\right)}
\end{equation}

In the following, we will first analyze model \eqref{MAINMAPNested2}, then \eqref{MAINMAPNested1}, and finally \eqref{MAINMAP}. Moreover, we will denote by $\mathcal{B}\left(\cdot\right)$ the basin of an attractor, i.e., the set of trajectories that converge to the attractor in the long run.

\section{Dynamics of the model without green technology progress, dynamically adjusted transition risk, and adjustment costs}\label{NestedVersion1}

If we drop the dependence of $\eta$ on $i$ for the sake of notation simplicity, the following results hold for model \eqref{MAINMAPNested2}.

\medskip

\begin{proposition}\label{Prop::ExpReplicatorStandard}
Consider $\beta>0$. Points $\bar{\eta}_{0}^{G}=0$ and $\bar{\eta}_{1}^{G}=1$ are equilibria of model \eqref{MAINMAPNested2}. Moreover, for
\begin{itemize}
\item[(1)] $\Pi_{G}>\Pi_{B}$, we have that $\bar{\eta}=1$ is globally stable in the sense that $\mathcal{B}\left(\bar{\eta}_{1}^{G}\right)=\left(0,1\right]$;
\item[(2)] $\Pi_{G}<\Pi_{B}$, we have that $\bar{\eta}=0$ is globally stable in the sense that $\mathcal{B}\left(\bar{\eta}_{0}^{G}\right)=\left[0,1\right)$;
\item[(3)] $\Pi_{G}=\Pi_{B}$, the segment $\left[0,1\right]$ is filled with equilibria that are marginally stable.
\end{itemize}
\end{proposition}

\medskip

\begin{figure}[!htbp] 
    \centering
    \includegraphics[scale=0.7]{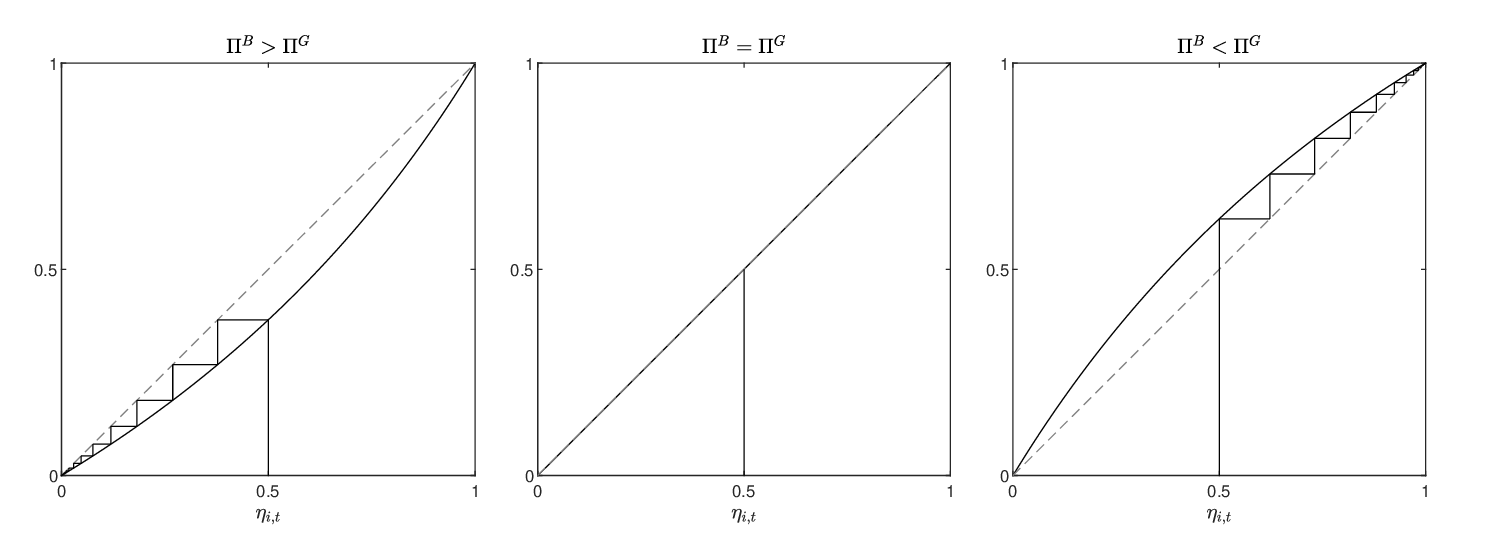}
    \caption{Staircase diagrams for model \eqref{MAINMAPNested2}. Panel (a), $\Pi_{G} = 0.5$ and $\Pi_{B} = 1$. Panel (b), $\Pi_{G} = 1$ and $\Pi_{B} = 1$. Panel (c), $\Pi_{G} = 1$ and $\Pi_{B} = 0.5$. Remaining parameters: $\beta = 1$.}
    \label{Figure2}
\end{figure}

The results given in Proposition \ref{Prop::ExpReplicatorStandard}, and their graphical representation using the staircase diagrams in Figure \ref{Figure2}, indicate that the green technology will be adopted in the long run if the profit generated by this technology is higher than the profit generated by the brown technology. Otherwise, the brown technology will be adopted in the long run. By threatening a stringent environmental regulation, such as a lump sum tax greater than the profit gap between brown and green technologies, policymakers can ensure a green transition even if the brown technology is more profitable.

\section{Dynamics of the model with adjustment costs but without green technology progress and dynamically adjusted transition risk}\label{NestedVersion2}

The following results hold for model \eqref{MAINMAPNested1}.

\medskip

\begin{proposition}\label{Prop::ExpReplicatorEvoluta}
Assume $C^{B},C^{G}\geq0$, either $C^{B}>0$ or $C^{G}>0$, and $\beta>0$. Moreover, let $\bar{\eta}_{0}^{i,G}=0$, $\bar{\eta}_{1}^{i,G}=1$ and
\begin{equation}
\bar{\eta}^{i,G}_{in} = \frac{\exp\left(\beta\left(\Pi^{B}-\Pi^{G}+C^{G}\right)\right)-1}{\exp\left(\beta\left(\Pi^{B}-\Pi^{G}+C^{G}\right)\right)-1+\exp\left(\beta\left(\Pi^{G}-\Pi^{B}+C^{B}\right)\right)-1}\text{.}
\end{equation}
We have the following:
\begin{itemize}
\item[(1)] for $\Pi_{B}\geq \Pi_{G}+C^{B}$, $\bar{\eta}_{0}^{i,G}$ and $\bar{\eta}_{1}^{i,G}$ are the equilibria, $\bar{\eta}_{1}^{i,G}$ is unstable, while $\bar{\eta}_{0}^{i,G}$ is locally asymptotically stable and $\mathcal{B}\left(\bar{\eta}_{0}^{i,G}\right)=\left[0,1\right)$;
\item[(2)] for $\Pi_{G}+C^{B}>\Pi_{B}>\Pi_{G}-C^{G}$, $\bar{\eta}_{0}^{i,G}$, $\bar{\eta}_{1}^{i,G}$ and $\bar{\eta}^{i,G}_{in}$ are the equilibria, $\bar{\eta}^{i,G}_{in}$ is unstable, while $\bar{\eta}_{1}^{i,G}$ and $\bar{\eta}_{0}^{i,G}$ are locally asymptotically stable and $\mathcal{B}\left(\bar{\eta}_{0}^{i,G}\right)=\left[0,\bar{\eta}^{i,G}_{in}\right)$ and $\mathcal{B}\left(\bar{\eta}_{1}^{i,G}\right)=\left(\bar{\eta}^{i,G}_{in},1\right]$;
\item[(3)] for $\Pi_{G}-C^{G}\geq \Pi_{B}$, $\bar{\eta}_{0}^{i,G}$ and $\bar{\eta}_{1}^{i,G}$ are the equilibria, $\bar{\eta}_{0}^{i,G}$ is unstable, while $\bar{\eta}_{0}^{i,G}$ is locally asymptotically stable and $\mathcal{B}\left(\bar{\eta}_{1}^{i,G}\right)=\left(0,1\right]$.
\end{itemize}
\end{proposition}

\medskip

\begin{figure}
    \centering
    \includegraphics[scale=0.75]{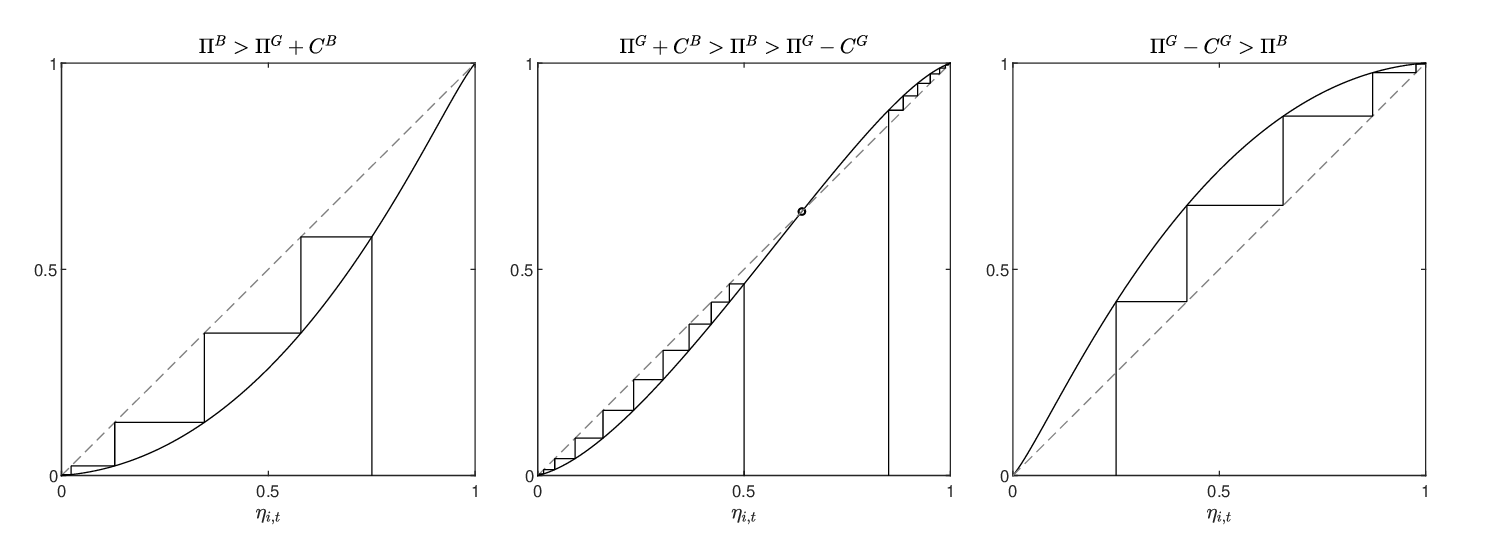}
    \caption{Staircase diagrams for model \eqref{MAINMAPNested1}. Panel (a), $\Pi_{B} = 1.3$. Panel (b), $\Pi_{B} = 1$. Panel (c), $\Pi_{B} = 0.6$. Remaining parameters: $\Pi_{G} = 0.95$; $C^{G} = 0.3$; $C^{B} = 0.3$ and $\beta = 4$.}
    \label{Figure1}
\end{figure}

The results given in Proposition \ref{Prop::ExpReplicatorEvoluta}, and their graphical representation using the staircase diagrams in Figure \ref{Figure1}, indicate that the green technology will be adopted in the long run if the profit generated by this technology, reduced by the cost of switching to the green technology, is higher than the profit generated by the brown technology. In contrast, the brown technology will be adopted in the long run if the profit generated by this technology, reduced by the cost of switching to the brown technology, is higher than the profit generated by the green technology. In all the other cases, both the probability of adopting a green technology and the probability of adopting a brown technology are positive. Comparing these results with those obtained for the model without switching costs, we can observe that in the presence of switching costs, a technology becomes the dominant technology only if its net profit (reduced by the switching cost) is higher than the gross profit (not reduced by the switching cost) of the other technology. In summary, switching costs increase the uncertainty about the technology that will be adopted. Moreover, the higher the cost to switching to a technology, the lower the probability that the technology will be adopted. Finally, due to the adjustment costs, there is the risk that the green technology will not be adopted even though it is the most profitable one. We also note that the intensity of choice affects on the value of the inner equilibrium $\bar{\eta}^{i,G}_{in}$ and thus the green transition. This does not happen when we consider a classical exponential replicator dynamics. In particular, we have the following.

\medskip

\begin{corollary}\label{Corr1}
Consider $\Pi_{G}+C^{B}>\Pi_{B}>\Pi_{G}-C^{G}$ and $\beta>0$. Then $\bar{\eta}^{i,G}_{in}$ is an unstable inner equilibrium. Moreover,
\begin{itemize}
\item[(1)] for $\beta\rightarrow 0$, we have
\begin{equation}
\lim_{\beta\rightarrow 0^{+}} \bar{\eta}^{i,G}_{in} = \bar{\eta}^{i,G}_{in,\left(\beta\rightarrow 0\right)} : = \frac{\Pi^{B}-\Pi^{G}+C^{G}}{C^{B}+C^{G}}\in\left(0,1\right)    
\end{equation}
and the larger the profit gap between a brown and a green technology, i.e. $\Pi^{B}-\Pi^{G}$, the larger the basin of attraction of equilibrium $\bar{\eta}_{0}^{i,G}=0$.
\item[(2)] for $\beta\rightarrow +\infty$, we have
\begin{equation}\label{betainfty}
\displaystyle
\lim_{\beta\rightarrow +\infty} \bar{\eta}^{i,G}_{in}  = \bar{\eta}^{i,G}_{in,\left(\beta\rightarrow +\infty\right)} : =
\left\{
\begin{array}{lll}
1     & \text{if} &  2\left(\Pi^{G}-\Pi^{B}\right) + C^{B} - C^{G} <0 \\
\\
0  & \text{if} & 2\left(\Pi^{G}-\Pi^{B}\right) + C^{B} - C^{G} >0
\end{array}
\right.
\end{equation}
and if the profit gap of remaining a green firm is higher than the profit gap of remaining a brown firm, i.e. $\Pi^{G}- \left(\Pi^{B}-C^{B}\right)>\Pi^{B}- \left(\Pi^{G}-C^{G}\right)$, then $\mathcal{B}\left(\bar{\eta}_{1}^{i,G}\right)$ becomes $\left(0,1\right]$, i.e., $\bar{\eta}_{1}^{i,G}$ is a global attractor. In contrast, if the profit gap of remaining a green firm is lower than the profit gap of remaining a brown firm, i.e., $\Pi^{G}- \left(\Pi^{B}-C^{B}\right)<\Pi^{B}- \left(\Pi^{G}-C^{G}\right)$, then $\mathcal{B}\left(\bar{\eta}_{0}^{i,G}\right)$ becomes $\left[0,1\right)$, i.e., $\bar{\eta}_{0}^{i,G}$ is a global attractor.
\item[(3)] if the profit gap of remaining a green firm is higher than the profit gap of remaining a brown firm, i.e., $\Pi^{G}- \left(\Pi^{B}-C^{B}\right)>\Pi^{B}- \left(\Pi^{G}-C^{G}\right)$, then the basin of attraction of $\bar{\eta}_{1}^{i,G}$ increases as the intensity of choice, parameter $\beta$, increases, otherwise it decreases.
\end{itemize}
\end{corollary}

\medskip

The results in Corollary \eqref{Corr1} indicate that the green transition also depends on the propensity of firms to switch to the more profitable technology. The research conducted so far provides a number of implications for economic policy. First of all, $\tau^{B}_{1}:=\max\left\{\Pi^{B}-\Pi^{G}+C^{G};0\right\}$ quantifies the minimum tax that must be imposed on the brown technology to ensure the green transition. Moreover, $\tau^{B}_{2}:=\max\left\{\Pi^{B}-\Pi^{G}-C^{B};0\right\}\left(<\tau^{B}_{1}\right)$ quantifies the minimum negative impact on the profit of brown production that a green regulation must have in order to create the necessary conditions for the green transition. Finally, the higher the intensity of the choice, parameter $\beta$, the lower the lump sum tax required to enforce the green transition. For $\beta\rightarrow +\infty$, for example, a green regulation that reduces the profit of a brown production by the amount $\tau^{B}_{3}:=\max\left\{\Pi^{B}-\Pi^{G}-\frac{C^{B}-C^{G}}{2};0\right\}\in\left[\tau^{B}_{2},\tau^{B}_{1}\right]$ is a sufficient condition for the green transition.

\section{Dynamics of the model with adjustment costs, green technology progress, and dynamically adjusted transition risk}\label{fullfledgedmodel}

\subsection{Overview and possible scenarios}

Model \eqref{MAINMAP} can be rewritten as
\begin{equation}\label{System::Main}
\begin{array}{lll}
\eta^{1,G}_{t+1} & = &  \eta^{1,G}_{t}\frac{\eta^{1,G}_{t}}{\eta^{1,G}_{t}+\left(1-\eta^{1,G}_{t}\right)\exp\left(\beta \left(a\eta^{2,G}_{t}+b-C^{B}\right)\right)} + \left(1-\eta^{1,G}_{t}\right) \frac{\eta^{1,G}_{t}}{\eta^{1,G}_{t}+\left(1-\eta^{1,G}_{t}\right)\exp\left(\beta\left(a\eta^{2,G}_{t}+b+C^{G}\right)\right)}  \\
\\
\eta^{2,G}_{t+1} & = &  \eta^{2,G}_{t}\frac{\eta^{2,G}_{t}}{\eta^{2,G}_{t}+\left(1-\eta^{2,G}_{t}\right)\exp\left(\beta \left(a\eta^{1,G}_{t}+b-C^{B}\right)\right)} + \left(1-\eta^{2,G}_{t}\right) \frac{\eta^{2,G}_{t}}{\eta^{2,G}_{t}+\left(1-\eta^{2,G}_{t}\right)\exp\left(\beta\left(a\eta^{1,G}_{t}+b+C^{G}\right)\right)}\text{,}
\end{array}
\end{equation}
where
\begin{equation}\label{a_and_b}
\begin{array}{lll}
a & = & \Pi^{BG}-\Pi^{GG} - \Pi^{BB} + \Pi^{GB} = \Pi^{BG} - \Pi^{BB} - \left(\Pi^{GG} - \Pi^{GB}\right) = \Pi^{BG}-\Pi^{GG} - b\\
\\
b & = & \Pi^{BB} - \Pi^{GB}\text{.}
\end{array}
\end{equation}

Model \eqref{System::Main} is a pair of two exponential replicator dynamics adjusted to account for the cost of switching strategies. These exponential replicator dynamics are proposed for the first time and represent a generalization of the classical exponential replicator dynamics, see, e.g., \citet{HofbauerSigmund2003}. Thus, the dynamics of model \eqref{System::Main} is also new and is characterized by peculiarities that are worth investigating on their own. These peculiarities and the main properties of the dynamics of the model are summarized in Lemma \ref{Lemma:MainProperties} in Appendix \ref{AppA}. Moreover, Appendix \ref{AppA} contains some numerical simulations that show the variety and complexity of the dynamics of model \eqref{System::Main}. These properties will be used to prove the results that follow, as well as their robustness.

Here we simply recall that the unit box $\mathcal{D}:=\left[0,1\right]^{2}$ is invariant, trajectories cannot exit the box, and its vertices are always equilibria of the model. Vertex $\eta_{11} = \left(1,1\right)$ represents the green transition, i.e., both firms adopt the green technology, vertex $\eta_{00}=\left(0,0\right)$ represents a failed green transition, i.e., both firms adopt the brown technology, while vertices $\eta_{01} = \left(0,1\right)$ and $\eta_{10} = \left(1,0\right)$ represent a partial green transition, i.e., one firm adopts the green technology and the other firm adopts the brown technology. These are the most relevant equilibria of the model, and their stability determines whether the green transition can occur or not. A necessary condition for the green transition is, for example, the local stability of the vertex equilibrium $\eta_{11}$. In addition to the vertex equilibria, there are four other possible equilibria at the edges of box $\left[0,1\right]^{2}$, given by $\left(0,\eta^{*}\right)$, $\left(\eta^{*},0\right)$, $\left(1,\eta^{+}\right)$, and $\left(\eta^{+},1\right)$, where $\eta^{*}$ and $\eta^{+}$ are defined in Appendix \ref{AppA}. However, we can also have equilibria inside the box $\left[0,1\right]^{2}$. These equilibria are called \emph{inner equilibria} in the following.

The instability of the inner equilibria can only be conjectured on the basis of the analytical results and the extensive numerical simulations perfomred and partially reported in Appendix \ref{AppA}. Instead, the same investigation conducted in Appendix \ref{AppA} emphasizes that $\left(0,\eta^{*}\right)$, $\left(\eta^{*},0\right)$, $\left(1,\eta^{+}\right)$, and $\left(\eta^{+},1\right)$ cannot be stable and the stability of the vertex equilibria can be determined by analytical arguments. These analytical conditions and the results in Lemma \ref{Lemma:MainProperties} in Appendix \ref{AppA} are then employed to investigate the possible scenarios with respect to the vertex equilibria. All the possible scenarios are described in the following proposition.

\medskip

\begin{proposition}\label{Prop:FullModelSce}
Assume $C^{B},C^{G}\geq0$, either $C^{B}>0$ or $C^{G}>0$, $\beta>0$, and $a\neq 0$. Consider $\eta^{*}$ and $\eta^{+}$ defined as in Lemma \ref{Lemma:MainProperties} in Appendix \ref{AppA} and consider a generic number $\bar{\eta}\in\left(0,1\right)$. Regarding the vertex equilibria $\eta_{00}=\left(0,0\right)$, $\eta_{10}=\left(1,0\right)$, $\eta_{11}=\left(1,1\right)$, and $\eta_{01}=\left(0,1\right)$, nine scenarios are possible:
\begin{itemize}
\item[1.] For $\Pi^{BB} > \Pi^{GB} - C^{G}$, $\Pi^{GG} > \Pi^{BG}-C^{B}$, $\Pi^{BG}>\Pi^{GG}-C^{G}$, and $\Pi^{GB}>\Pi^{BB} -C^{B}$, equilibria $\eta_{00}$, $\eta_{11}$, $\eta_{10}$, and $\eta_{01}$ are locally asymptotically stable, while equilibria $\left(0,\eta^{*}\right)$, $\left(\eta^{*},0\right)$, $\left(1,\eta^{+}\right)$ and $\left(\eta^{+},1\right)$ are saddles. There are no other equilibria at the edges, but there is at least one unstable inner equilibrium $\left(\bar{\eta},\bar{\eta}\right)$. 
\item[2.] For $\Pi^{BB} > \Pi^{GB} - C^{G}$, $\Pi^{GG} > \Pi^{BG}-C^{B}$, $\Pi^{BG}>\Pi^{GG}-C^{G}$, and $\Pi^{GB}<\Pi^{BB} -C^{B}$, equilibria $\eta_{00}$ and $\eta_{11}$ are locally asymptotically stable, while equilibria $\eta_{10}$, $\eta_{01}$, $\left(1,\eta^{+}\right)$, and $\left(\eta^{+},1\right)$ are saddles. There are no other equilibria at the edges, but there is at least one unstable inner equilibrium $\left(\bar{\eta},\bar{\eta}\right)$. 
\item[3.] For $\Pi^{BB} > \Pi^{GB} - C^{G}$, $\Pi^{GG} > \Pi^{BG}-C^{B}$, $\Pi^{BG}<\Pi^{GG}-C^{G}$, and $\Pi^{GB}<\Pi^{BB} -C^{B}$, equilibria $\eta_{00}$ and $\eta_{11}$ are locally asymptotically stable, while equilibria $\eta_{10}$ and $\eta_{01}$ are repellors. There are no other equilibria at the edges, but there is at least one unstable inner equilibrium $\left(\bar{\eta},\bar{\eta}\right)$. 
\item[4.] For $\Pi^{BB} > \Pi^{GB} - C^{G}$, $\Pi^{GG} < \Pi^{BG}-C^{B}$, $\Pi^{BG}>\Pi^{GG}-C^{G}$, and $\Pi^{GB}>\Pi^{BB} -C^{B}$, equilibria $\eta_{00}$, $\eta_{10}$, and $\eta_{01}$ are locally asymptotically stable, equilibrium $\eta_{11}$ is a repellor, while equilibria $\left(0,\eta^{*}\right)$ and $\left(\eta^{*},0\right)$ are saddles. There are no other equilibria at the edges. 
\item[5.] For $\Pi^{BB} > \Pi^{GB} - C^{G}$, $\Pi^{GG} < \Pi^{BG}-C^{B}$, $\Pi^{BG}>\Pi^{GG}-C^{G}$ and $\Pi^{GB}<\Pi^{BB} -C^{B}$, equilibrium $\eta_{00}$ is locally asymptotically stable, equilibria $\eta_{10}$ and $\eta_{01}$ are saddles, while equilibrium $\eta_{11}$ is unstable. There are no other equilibria on the edges. 
\item[6.] For $\Pi^{BB} > \Pi^{GB} - C^{G}$, $\Pi^{GG} > \Pi^{BG}-C^{B}$, $\Pi^{BG}<\Pi^{GG}-C^{G}$, and $\Pi^{GB}>\Pi^{BB} -C^{B}$, equilibria $\eta_{00}$ and $\eta_{11}$ are locally asymptotically stable, while equilibria $\eta_{10}$, $\eta_{01}$, $\left(0,\eta^{*}\right)$, and $\left(\eta^{*},0\right)$ are saddles. There are no other equilibria at the edges, but there is at least one unstable inner equilibrium $\left(\bar{\eta},\bar{\eta}\right)$. 
\item[7.] For $\Pi^{BB} < \Pi^{GB} - C^{G}$, $\Pi^{GG} > \Pi^{BG}-C^{B}$, $\Pi^{BG}>\Pi^{GG}-C^{G}$, and $\Pi^{GB}>\Pi^{BB} -C^{B}$, equilibria $\eta_{11}$, $\eta_{10}$, and $\eta_{01}$ are locally asymptotically stables, equilibrium $\eta_{00}$ is a repellor, and equilibria $\left(1,\eta^{+}\right)$ and $\left(\eta^{+},1\right)$ are saddle. There are no other equilibria at the edges. 
\item[8.] For $\Pi^{BB} < \Pi^{GB} - C^{G}$, $\Pi^{GG} < \Pi^{BG}-C^{B}$, $\Pi^{BG}>\Pi^{GG}-C^{G}$, and $\Pi^{GB}>\Pi^{BB} -C^{B}$, equilibria $\eta_{10}$ and $\eta_{01}$ are locally asymptotically stable, while equilibria $\eta_{00}$ and $\eta_{11}$ are repellors. There are no other equilibria at the edges. 
\item[9.] For $\Pi^{BB} < \Pi^{GB} - C^{G}$, $\Pi^{GG} > \Pi^{BG}-C^{B}$, $\Pi^{BG}<\Pi^{GG}-C^{G}$, and $\Pi^{GB}>\Pi^{BB} -C^{B}$, equilibrium $\eta_{11}$ is locally asymptotically stable, equilibria $\eta_{10}$ and $\eta_{01}$ are saddles, and equilibrium $\eta_{00}$ is unstable. There are no other equilibria at the edges. 
\end{itemize}
\end{proposition}

\medskip

The green transition can only occur in some of the nine scenarios identified in Proposition \ref{Prop:FullModelSce}, namely in Scenarios 1, 2, 3, 6, 7, and 9. In these scenarios, vertex $\eta_{11}$ is stable, but the possibility of undergoing a green transition depends on the initial conditions. A basin-of-attraction analysis allows us to determine the set of initial conditions for which the green transition takes place. This set of initial conditions is depicted in green in the nine state spaces shown in Figure \ref{Figure19}, corresponding to the nine scenarios identified in Proposition \ref{Prop:FullModelSce}. In addition to the green region, representing the basin of attraction of the vertex equilibrium $\eta_{11}$, we can observe three other regions. The brown region is associated with a failed green transition. This is the basin of attraction of the vertex equilibrium $\eta_{00}$. The yellow and blue regions are associated with a partial green transition. These are the basins of attraction of the vertex equilibria $\eta_{10}$ and $\eta_{01}$, respectively. Clearly, these colored regions represent the set of trajectories that converge to a specific equilibrium. With the exception of Scenario 9, where the green transition is certain and $\eta_{11}$ is a global attractor, the green transition is possible in other scenarios, but its realization is conditioned on a starting point in the green region.

In Figure \ref{Figure19}, the parameters are chosen for exploratory purposes. With the help of basins of attraction, we aim to obtain a representation of the possible global dynamics of each possible scenario highlighted in Proposition \ref{Prop:FullModelSce}. The robustness of these global dynamics is confirmed by an extensive numerical investigation, only partially reported in Appendix \ref{AppA} due to space constraints. In light of this, we use the numerical examples of Figure \ref{Figure19} to comment on the implications of the various scenarios and on the economic conditions that characterize them. We follow the order of Proposition \ref{Prop:FullModelSce}, that is, we start with Scenario 1 and end with Scenario 9. At the end of this review, we analyze possible green economic policies to avoid scenarios that hinder the green transition.

\begin{figure}[!t]
    \begin{center}
    \includegraphics[scale=0.55]{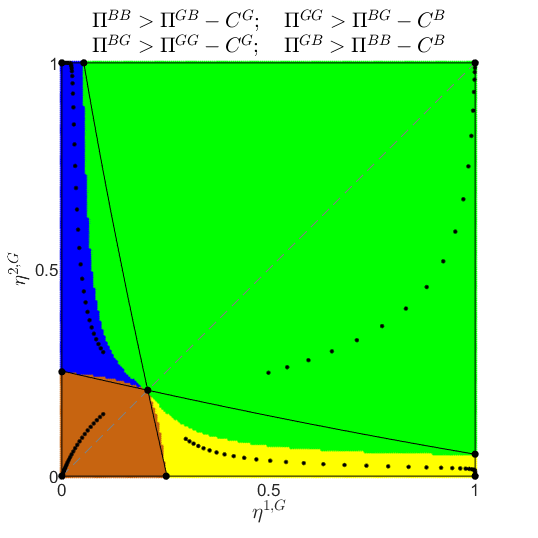}
    \includegraphics[scale=0.55]{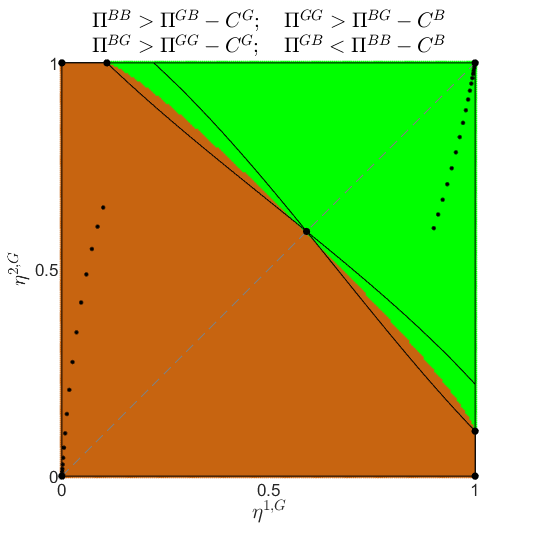}
    \includegraphics[scale=0.55]{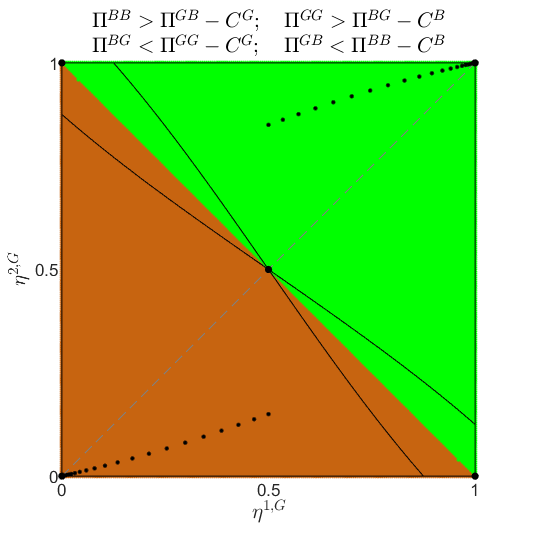}\\
    \includegraphics[scale=0.55]{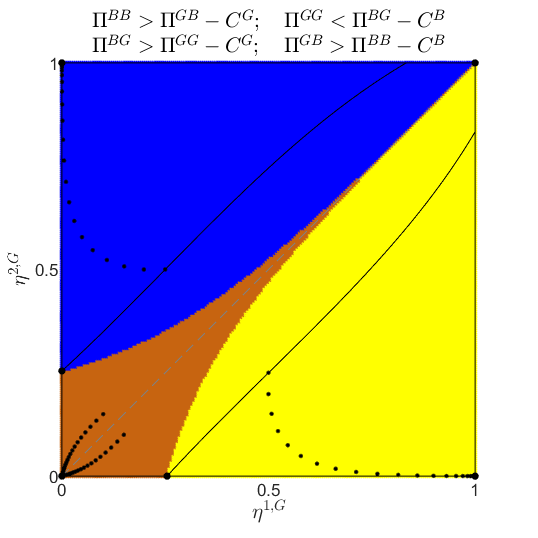}
    \includegraphics[scale=0.55]{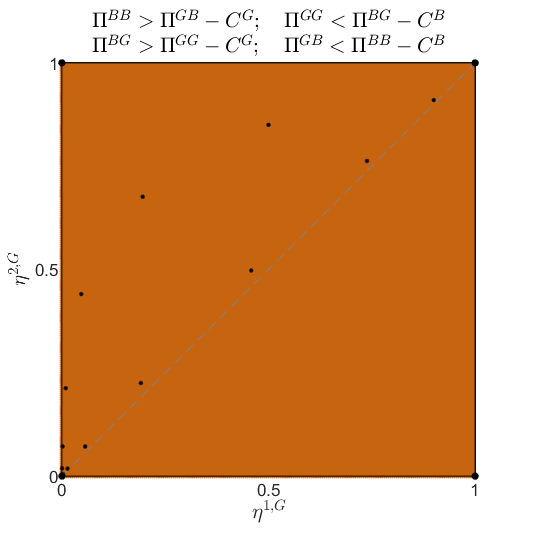}
    \includegraphics[scale=0.55]{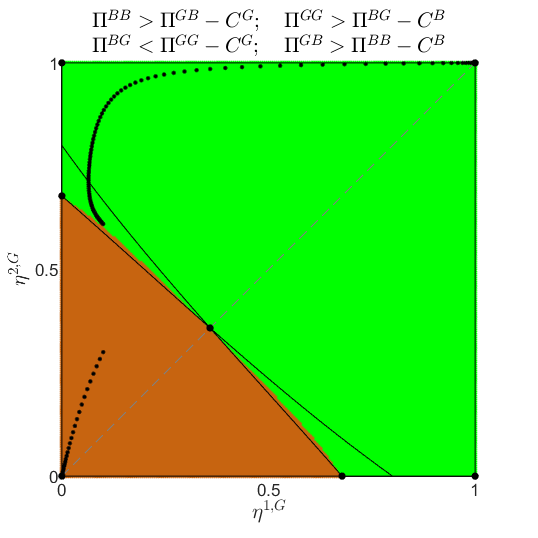}\\
    \includegraphics[scale=0.55]{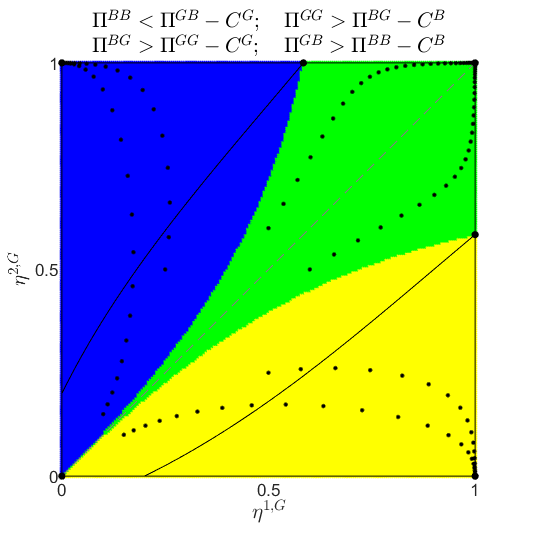}
    \includegraphics[scale=0.55]{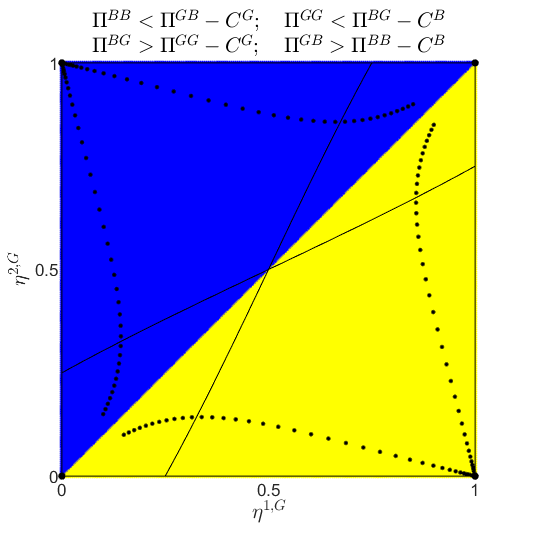}
    \includegraphics[scale=0.55]{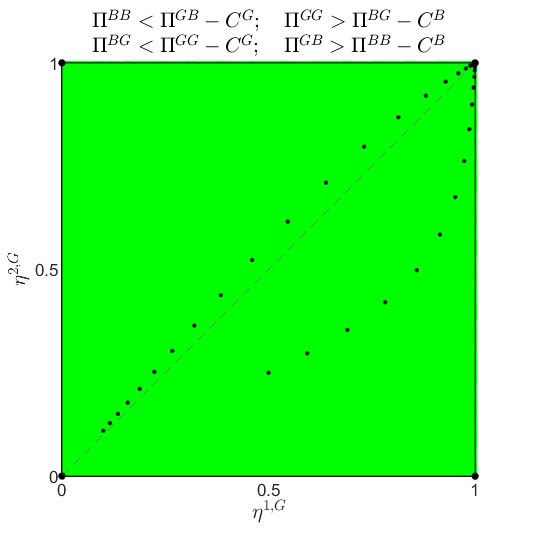}
        \caption{\small Top-left panel (Scenario 1), $\Pi^{BB} = 2.2$; $\Pi^{GB} = 2.3$; $C^{G} = 0.3$;
$\Pi^{GG} = 2.75$; $\Pi^{BG} = 2.5$; $C^{B} = 0.4$. Top-middle panel (Scenario 2), $\Pi^{BB} = 2.4$; $\Pi^{GB} = 2.2$; $C^{G} = 0.3$; $\Pi^{GG} = 2.75$; $\Pi^{BG} = 2.5$; $C^{B} = 0.1$. Top-right panel (Scenario 3), $\Pi^{BB} = 2.2$; $\Pi^{GB} = 2.05$; $C^{G} = 0.2$; $\Pi^{GG} = 2.75$; $\Pi^{BG} = 2.5$; $C^{B} = 0.1$. Middle-left panel (Scenario 4), $\Pi^{BB} = 2.2$; $\Pi^{GB} = 2.3$; $C^{G} = 0.3$; $\Pi^{GG} = 2.0$; $\Pi^{BG} = 2.5$; $C^{B} = 0.4$. Middle-middle panel, (Scenario 5), $\Pi^{BB} = 2$; $\Pi^{GB} = 1$; $C^{G} = 0.5$; $\Pi^{GG} = 1$; $\Pi^{BG} = 2.5$; $C^{B} = 0.4$. Middle-right panel (Scenario 6) $\Pi^{BB} = 2.2$; $\Pi^{GB} = 2.2$; $C^{G} = 0.2$; $\Pi^{GG} = 2.75$; $\Pi^{BG} = 2.5$; $C^{B} = 0.1$. Bottom-left panel (Scenario 7), $\Pi^{BB} = 1.9$; $\Pi^{GB} = 2.3$; $C^{G} = 0.3$; $\Pi^{GG} = 2.4$; $\Pi^{BG} = 2.5$; $C^{B} = 0.4$. Bottom-middle panel (Scenario 8), $\Pi^{BB} = 2.1$; $\Pi^{GB} = 2.3$; $C^{G} = 0.1$; $\Pi^{GG} = 2.3$; $\Pi^{BG} = 2.5$; $C^{B} = 0.1$. Bottom-right panel (Scenario 9), $\Pi^{BB} = 2$; $\Pi^{GB} = 2.3$; $C^{G} = 0.2$; $\Pi^{GG} = 2.75$; $\Pi^{BG} = 2.5$; $C^{B} = 0.4$. Remaining parameters: $\beta=1$.}\label{Figure19}
    \end{center}
\end{figure}

Scenario 1 is characterized by conservatism, since it is never convenient to change technologies. Indeed, $\Pi^{BB} > \Pi^{GB} - C^{G}$ indicates that it is economically convenient to remain brown when the other firm is brown. $\Pi^{GG} > \Pi^{BG}-C^{B}$ indicates that it is better to remain green when the other firm is green. $\Pi^{BG}>\Pi^{GG}-C^{G}$ indicates that it is profitable to remain brown when the other firm is green. Finally, $\Pi^{GB}>\Pi^{BB} -C^{B}$ indicates that it is economically convenient to remain green when the other firm is brown. Under these economic conditions, all the vertex equilibria are stable, i.e., we can expect green vs green, green vs brown, brown vs green, and brown vs brown outcomes. This is a kind of paradox, since conservatism generates scenarios characterized by production uncertainty. All in all, this scenario foresees three possible developments: A green transition if the orbit converges to equilibrium $\eta_{11}$; a partial green transition if the orbit converges to either equilibrium $\eta_{01}$ or equilibrium $\eta_{10}$; and a failed green transition if the orbit converges to equilibrium $\eta_{00}$. To achieve a green transition in such a scenario, the starting conditions are therefore important. Numerical simulations, not reported here, indicate the following. Increasing the profit of the green firm, but remaining in the parameter configurations of Scenario 1, increases the set of trajectories leading to the green transition, while decreasing the set of trajectories leading to a missing green transition and the set of trajectories leading to a partial green transition. In the numerical example of Figure \ref{Figure19}-(a), the set of trajectories that lead to a green transition is already relatively large compared to the set of trajectories that lead to a partial or missing green transition. Hence, the risk of a failed green transition is low, albeit present.

Scenario 2 is characterized by the same conservatism as Scenario 1, with the only exception that it is economically convenient to switch to a brown technology when the other firm is brown, as indicated by the condition $\Pi^{BB} -C^{B}>\Pi^{GB}$. Obviously, this modified profit condition excludes the partial green transition, where one firm adopts a green technology and the other adopts a brown technology. The remaining possibilities are therefore a green transition and a failed green transition. In the numerical example of Figure \ref{Figure19}-(b), the set of orbits leading to a failed green transition (brown region) is wider than the set of orbits leading to a green transition (green region). Hence, the risk of a failed green transition is high. At least for a large variety of parameter configurations characterizing Scenario 2, we expect a similar situation, i.e., the probability of a failed green transition is higher than the probability of a green transition. Technically speaking, this is due to the position of the equilibria at the edges of the box $\left[0,1\right]$.\footnote{In Scenario 2, the stable manifolds of $\left(1,\eta^{+}\right)$ and $\left(\eta^{+},1\right)$ are the border of the green and brown regions. These saddle equilibria lie at the edges where $\eta_{11}$ lies. Hence, they are closer to $\eta_{11}$ than to $\eta_{00}$. It follows that we expect the border of the brown/green region to be closer to $\eta_{11}$ than to $\eta_{00}$, i.e., the brown region is wider.} Economically speaking, this is due to the fact that it is economically convenient to switch to the brown technology when the other firm is brown, while it is never convenient to switch technologies in all other cases.

The profit conditions that characterize Scenario 3 are similar to those of Scenario 2, with the only exception that it is economically convenient to switch to the green technology when the other firm is green ($\Pi^{GG}-C^{G}>\Pi^{BG}$). As a result, Scenario 3 is qualitatively equivalent to Scenario 2. Hence, a green transition and a failed green transition are the two possibilities. Quantitatively speaking, the green region is wider in Scenario 3 than in Scenario 2. Hence, the likelihood of a green transition is also higher.\footnote{This is due to the saddle equilibria at the edges, whose stable manifolds are the boundaries of the green and brown regions. Excluding the vertices, the edge equilibria are present in Scenario 2 and absent in Scenario 3. In Scenario 3, the boundary of the green and brown regions is given by the stable manifold of the saddle $\left(\bar{\eta},\bar{\eta}\right)$ connected to vertices $\eta_{10}$ and $\eta_{0,1}$, dividing the state space into two regions of equal width, the brown and the green. In Scenario 2, instead, there are two saddle equilibria at the edges, i.e., $\left(1,\eta^{+}\right)$ and $\left(\eta^{+},1\right)$, the stable manifolds of which are the boundaries of the green and brown regions. These saddle equilibria lie at the edges where $\eta_{11}$ lies. Hence, they are closer to $\eta_{11}$ than to $\eta_{00}$. It follows that we expect that the boundary of the brown/green regions to be closer to $\eta_{11}$ than to $\eta_{00}$, i.e., the brown region is wider. For more details, see Lemma \ref{Lemma:MainProperties} in Appendix \ref{AppA}.} For instance, compare the green and brown regions in Figure \ref{Figure19}-(c), where the size of the green region coincides with the size of the brown region. This implies that the probability of a green transition is equal to the probability of a failed green transition.\footnote{The equal size of the green and brown regions follows from the fact that the borders of these two regions are the stable manifold of the inner equilibrium, which in the particular case of Figure \ref{Figure19}-(c) is $\left(0.5,0.5\right)$, and from the symmetry of the model with respect to the diagonal $\eta^{1,G}=\eta^{2,G}$ of the square box $\left[0,1\right]^{2}$. See Lemma\ref{Lemma:MainProperties}-(P10) in Appendix \ref{AppA}, for the analytical conditions required to have the inner equilibrium $\left(0.5,0.5\right)$.}

Scenario 4 differs from the conservatism of Scenario 1 and from Scenarios 2 and 3 because it is economically convenient to switch to the brown technology when the other firm is green, as indicated by the condition $\Pi^{BG}-C^{B}>\Pi^{GG}$. As a result, the green transition cannot take place in Scenario 4, while it is possible in Scenarios 1, 2, and 3. Moreover, Scenario 4 differs from Scenarios 2 and 3 because it is economically convenient to remain green when the other firm is brown ($\Pi^{GB}>\Pi^{BB}-C^{B}$), and it differs from Scenario 3, because it is economically convenient to remain brown when the other firm is green ($\Pi^{BG}>\Pi^{GG}-C^{G}$). Consequently, a partial green transition is possible in Scenario 4, as it is in Scenario 1, while it is not possible in Scenarios 2 and 3. The likelihood of a failed green transition is a common element of Scenarios 1, 2, 3, and 4. As can be seen in Figure \ref{Figure19}-(d), the brown region yields trajectories responsible for a failed green transition, while the blue and yellow regions yield trajectories responsible for a partial green transition.

In Scenario 5, a green transition is not even partially feasible. In this scenario, it is always economically convenient to remain brown, both when the other firm is brown ($\Pi^{BB} > \Pi^{GB} - C^{G}$) and when it is green ($\Pi^{BG}>\Pi^{GG}-C^{G}$). Moreover, it is always economically convenient to switch to a brown technology, both when the other firm is brown ($\Pi^{BG}-C^{B}>\Pi^{GG}$) and when it is green ($\Pi^{BB} -C^{B}>\Pi^{GB}$). As shown in Figure \ref{Figure19}-(e), all trajectories end up in the equilibrium where both firms adopt a brown technology. Indeed, the brown region covers the entire state space except for subsets of zero measure.\footnote{By an abuse of notation, we denote as global attractors those invariant sets that attract the entire state space, with the possible exception of subsets of zero measure.} This is the worst possible scenario.

Apart from the fact that it is economically convenient to remain brown when the other firm is brown, i.e. $\Pi^{BB} > \Pi^{GB} - C^{G}$, Scenario 6 is the opposite of Scenario 5. Indeed, it is economically convenient to remain green when the other firm is green ($\Pi^{GG} > \Pi^{BG}-C^{B}$), it is economically convenient to switch to green when the other firm is green ($\Pi^{GG}-C^{G}>\Pi^{BG}$), and it is economically convenient to remain green when the other firm is brown ($\Pi^{GB}>\Pi^{BB} -C^{B}$). Qualitatively speaking, Scenario 6 is similar to Scenarios 2 and 3, where a green transition and a failed green transition are the only possible long-run outcomes. Quantitatively speaking, we conjecture that the basin of attraction of the green transition is wider in Scenario 6, which is also confirmed by Figure \ref{Figure19}-(f). Hence, the probability of a green transition is likely to be higher in Scenario 6 compared to Scenarios 2 and 3.\footnote{This depends on the saddle equilibria at the edges, whose stable manifold separates the green and brown regions. Being located at the edges where equilibrium $\eta_{00}$ is also located, in Scenario 6, these saddle equilibria are closer to $\eta_{00}$ than to $\eta_{11}$. Hence, we expect the boundary of the green/brown regions to be closer to $\eta_{00}$ than to $\eta_{11}$, and the basin of attraction of $\eta_{11}$ to be wider than that of the brown region. For more details, see Lemma \ref{Lemma:MainProperties} in Appendix \ref{AppA}.}

As in Scenario 6, in Scenario 7 it is economically convenient to remain green when the other firm is green ($\Pi^{GG} > \Pi^{BG}-C^{B}$), and it is economically convenient to remain green when the other firm is brown ($\Pi^{GB}>\Pi^{BB} -C^{B}$). Unlike Scenario 6, in Scenario 7, it is economically convenient to switch to green when the other firm is brown ($ \Pi^{GB} - C^{G}>\Pi^{BB}$) and it is more profitable to remain brown when the other firm is green ($\Pi^{BG}>\Pi^{GG}-C^{G}$). As a result, the green transition remains a possible outcome, but the partial green transition substitutes the failed green transition. See Figure \ref{Figure19}-(g).

Scenario 8 is similar to Scenario 7, except that it is economically convenient to switch to the brown technology when the other firm adopts a green technology ($\Pi^{BG}-C^{B}>\Pi^{GG}$). The effect of this is a partial green transition, i.e., the green transition is no longer possible. See Figure \ref{Figure19}-(h).

Finally, Scenario 9 is the desirable scenario in which the green transition is the unique possible outcome. In fact, in Figure \ref{Figure19}-(i), the green region covers the entire state space, except for subsets of zero measure. In Scenario 9, it is always economically convenient to remain green, both when the other firm is brown ($\Pi^{GG} > \Pi^{BG}-C^{B}$) and when it is green ($\Pi^{GB}>\Pi^{BB} -C^{B}$). Moreover, it is always economically convenient to switch to green, both when the other firm is brown ($\Pi^{GB} - C^{G}>\Pi^{BB}$) and when it is green ($\Pi^{GG}-C^{G}>\Pi^{BG}$). These are the unique (economic) conditions that ensure the green transition. Indeed, the green transition is only certain in Scenario 9, as can be seen in Figure \ref{Figure19}.

\subsection{Effects of dynamically adjusted transition risk and green technology progress}

Having analyzed all nine possible scenarios identified in Proposition \ref{Prop:FullModelSce}, our next goal is to isolate the negative effects of a dynamically adjusted transition risk on the green transition and the positive effects of green technology progress. To achieve this goal, the next corollary identifies the scenarios that emerge due to a dynamically adjusted transition risk and those that emerge due to green technology progress.

\medskip

\begin{corollary}\label{Corr2}
Note that:
\begin{itemize}
\item[(1)] if $\Pi^{GG}=\Pi^{GB} = \Pi^{G}$ and $\Pi^{BG}=\Pi^{BB} = \Pi^{B}$, then only Scenarios 1, 5,  and 9 are possible (neither green technology progress nor dynamically adjusted transition risk);
\item[(2)] if $\Pi^{GG}= \Pi^{GB} = \Pi^{G}$ and $\Pi^{BB}<\Pi^{BG}$, then Scenarios 2, 3, and 6 are not possible (dynamically adjusted transition risk but no green technology progress);
\item[(3)] if $\Pi^{BB} = \Pi^{BG} = \Pi^{B}$ and $\Pi^{GG}>\Pi^{GB}$, then Scenarios 4, 7, and 8 are not possible (green technology progress but no dynamically adjusted transition risk);
\item[(4)] if $\Pi^{BG}>\Pi^{BB}$ and $\Pi^{GG}>\Pi^{GB}$, then all nine scenarios of Proposition \ref{Prop:FullModelSce} are possible.
\end{itemize}
\end{corollary}

\medskip

Excluding any form of dynamically adjusted transition risk and green technology progress, Corollary \ref{Corr2}-(1) indicates that there are only three possible scenarios. First, Scenario 5, in which the green transition does not occur at all. The dynamics along the diagonal $\eta^{1,G}=\eta^{2,G}$ of the unit box $\mathcal{D}$ is consistent with the dynamics of the one-dimensional model \eqref{MAINMAPNested1} in the case of Proposition \ref{Prop::ExpReplicatorEvoluta}-(1). Second, Scenario 1, in which all outcomes are possible: a green transition, a partial green transition, and a failed green transition. The dynamics along the diagonal $\eta^{1,G}=\eta^{2,G}$ of the unit box $\mathcal{D}$ is consistent with the dynamics of the one-dimensional model \eqref{MAINMAPNested1} in the case of Proposition \ref{Prop::ExpReplicatorEvoluta}-(2). Third, Scenario 9, in which the green transition is achieved. The dynamics along the diagonal $\eta^{1,G}=\eta^{2,G}$ of the unit box $\mathcal{D}$ is consistent with the dynamics of the one-dimensional model \eqref{MAINMAPNested1} in the case of Proposition \ref{Prop::ExpReplicatorEvoluta}-(3). Note that by increasing transition risk, i.e., lowering the expected profit of a brown firm by, for example, threatening to impose a lump-sum tax on brown production, Scenario 5 can be transferred to Scenario 1 and Scenario 1 to Scenario 9. Hence, by increasing transition risk, it is possible to achieve a green transition, at least in the absence of green technology progress and when the transition risk is not dynamically adjusted based on the level of the greenness of the industry.

With respect to a dynamically adjusted transition risk, Corollary \ref{Corr2}-(2) indicates that it is consistent with Scenarios 1, 5, and 9. Therefore, despite the negative impact of the dynamic adjustment of the transition risk, it does not prevent the green transition and the most favorable scenario, i.e., Scenario 9, in which the green transition is certain. Moreover, Corollary \ref{Corr2}-(2) indicates that a dynamically adjusted transition risk is responsible for Scenarios 4, 7, and 8. Considering $\Pi^{BB}$ fixed, we enhance the dynamic adjustment of the transition risk, i.e., we increase the dependence of the transition risk on the number of brown firms. Hence, $\Pi^{BG}$ increases and we can only move from Scenario 1 to Scenario 4 and from Scenario 9 to either Scenario 7 or Scenario 8. Scenario 5 remains Scenario 5. Assuming green technology progress, we observe that if we dynamically adjust the transition risk by reducing it as the level of greenness of the industry increases, we can also move from Scenarios 2 and 3 to Scenario 5 and from Scenario 6 first to Scenario 1 and then to Scenario 4. Comparing these possible changes of scenarios using Figure \ref{Figure19}, we can observe that shifting from a transition risk to a dynamically adjusted transition risk reduces the chances of a green transition.

Transition risk is a policy tool to foster the green transition, but the research conducted so far suggests that an intensity of the transition risk that is inversely proportional to the level of greenness of the industry can be an obstacle. To further emphasize this aspect, we note that in the case of a transition risk that is not dynamically adjusted, i.e., $\Pi^{BB}=\Pi^{BG}$, it cannot be that it is both economically convenient to remain brown when the other firm is brown ($\Pi^{BB}<\Pi^{GB}-C^{G}$), and economically convenient to remain brown when the other firm is green ($\Pi^{BG}>\Pi^{GG}-C^{G}$). For example, consider Scenario 7, in which $\Pi^{BB}<\Pi^{BG}$, the green transition is either partial or total. Assume that transition risk is increased only when one firm is brown in such a way that $\Pi^{BB}=\Pi^{BG}<\Pi^{GG}-C^{G}<\Pi^{GB}-C^{G}$. Then, from Scenario 7, we move to Scenario 9, in which the green transition occurs. In contrast, by reducing the transition risk when both firms are brown, such that $\Pi^{BB}=\Pi^{BG}<\Pi^{GG}-C^{G}$, we move from Scenario 7 to Scenario 1. In this way, policymakers run the risk of making a failed green transition possible.

It is also worth noting that in the case of a transition risk that is not dynamically adjusted, it cannot be economically convenient to switch to a brown technology when the other firm is green ($\Pi^{GG}<\Pi^{BG}-C^{B}$) and, at the same time, it cannot be economically convenient to remain green when the other firm is brown ($\Pi^{GB}>\Pi^{BB}-C^{B}$). Consider Scenario 4, in which $\Pi^{BB}<\Pi^{BG}$ is required and transition risk is either absent or partially implemented. Assume that transition risk is increased when only one firm is brown, and is equal to the transition risk when both firms are brown. Then, $\Pi^{GG}>\Pi^{GB}>\Pi^{BG}-C^{B}=\Pi^{BB}-C^{B}$ and from Scenario 4 we move to Scenario 1, in which the green transition becomes a possibility. If we consider that we are in Scenario 8 and do the same, then we move to Scenario 9. Indeed, we will have $\Pi^{GG}>\Pi^{GB}>\Pi^{BG}-C^{B}=\Pi^{BB}-C^{B}$, but also $\Pi^{BB}=\Pi^{BG}<\Pi^{GB}-C^{G}<\Pi^{GG}-C^{G}$. That is, we move from a partial green transition (Scenario 8) to a fully green transition (Scenario 9).

In summary, if we remove the dynamic adjustment of transition risk by increasing it when there is only one brown firm in the industry, we move from Scenario 4 to Scenario 1 and from Scenarios 7 and 8 to Scenario 9. That is, the possibilities of a green transition increase. Instead, if we reduce the transition risk in the case of two brown firms to the level of transition risk in the case of only one brown firm, we move from Scenarios 4 and 8 to Scenario 5 and from Scenario 7 to Scenario 1. Hence, the risk of a failed green transition is either certain or increases.

These results emphasize that increasing the threat of introducing restrictions on brown production, e.g., in the form of a lump-sum tax, even if not all firms are brown, helps to achieve the green transition. In contrast, reducing transition risk hinders the green transition. To foster the green transition, transition risk should be kept at a high level, which we quantify in $\max\left\{\hat{\Pi}^{B}-\Pi^{GB}+C^{G};0\right\}$, where $\hat{\Pi}^{B}$ is the profit of a brown firm in the absence of transition risk. However, this is not enough. It should be kept at a high level regardless the number of brown firms in the industry, i.e., $\Pi^{BB} = \Pi^{GB} < \Pi^{GB} -C^{G}$.

With respect to green technology progress, Corollary \ref{Corr2}-(3) indicates that it is consistent with Scenarios 1, 5, and 9. Therefore, despite the positive impact of green technology progress on the green transition, it does not prevent the failure of the green transition and the worst scenario in which a failed green transition is certain. Moreover, Corollary \ref{Corr2}-(3) indicates that green technology progress is responsible for Scenarios 2, 3, and 6. Assuming that $\Pi^{GB}$ is fixed and that green technology progress increases, the only effect is to increase $\Pi^{GG}$. Hence, we can only move from Scenario 1 to Scenario 6 and from Scenario 5 to either Scenario 2 or Scenario 3, while from Scenario 9 we remain in Scenario 9. Assuming a transition risk that is dynamically adjusted, we observe that with increasing green technology progress, we can also move from Scenario 4 to one among Scenarios 1 and 6. Moreover, a transfer from Scenario 7 to Scenario 9 and from Scenario 8 to either Scenario 1 or Scenario 9 is possible. Comparing these possible changes of scenario using Figure \ref{Figure19}, we can observe that the chances of a green transition increase as green technology progresses.

With respect to the combined effects of green technology progress and a dynamically adjusted transition risk, Corollary \ref{Corr2}-(4) indicates that each of the nine scenarios of Proposition \ref{Prop:FullModelSce} are consistent with any level of benefits from green technology progress and any level of disadvantages caused by dynamically adjusting transition risk. This emphasizes that it is not possible to exclude some scenarios when $a>0$, i.e., when the disadvantages associated with a dynamically adjusted transition risk ($\Pi^{BG} - \Pi^{BB}>0$) are greater than the benefits of green technology progress ($\Pi^{GG} - \Pi^{GB}>0$) or when $a<0$. This is because the magnitude of the benefits due to green technology progress and the magnitude of the cost of a dynamically adjusted transition risk do not directly affect the relative benefits, which determine the evolutionary selection process. In fact, the decision on which the technology to adopt is based on the comparison between $\Pi^{BG}$ and either $\Pi^{GG}+C^{B}$ or $\Pi^{GG}-C^{G}$, or the comparison between $\Pi^{GB}$ and either $\Pi^{BB}-C^{B}$ or $\Pi^{BB}+C^{G}$. Indirectly, however, the sign of $a$ has an impact on these relative performances when combined with some green economic policies. For example, $a>0$ implies that it cannot be $\Pi^{GG}>\Pi^{BG}$ and at the same time $\Pi^{BB}>\Pi^{GB}$. Subsidizing the adoption of the green technology, i.e., $C^{G}=0$, Scenarios 3 and 6 are then excluded if $a>0$, while Scenarios 7 and 8 are excluded if $a<0$.

All in all, both transition risk and green technology progress help to avoid the risk of a failed green transition. Despite green technology progress, transition risk is the result of a green economic policy. A well-known and widely used way to increase transition risk and put pressure on brown firms is the threat of introducing a tax on brown production. Adding a tax on the use of a brown technology, we have the following results: 1)  With a tax of the type $\Pi^{BB}-\tau^{B}$ and $\Pi^{BG}-\tau^{B}$, we can exclude all scenarios except Scenario 9 as long as $\tau^{B}$ is sufficiently high, i.e., we can ensure the green transition. 2) With a tax of the type $\Pi^{BB}-\tau^{B}$, we can exclude Scenarios 1 to 6 as long as $\tau^{B}$ is sufficiently high, i.e., we can avoid the situation where both firms adopt a brown technology.

These results confirm that an increase in transition risk reduces the risk of a failed green transition, and to ensure the green transition, the intensity of transition risk must be high regardless of the number of brown firms that populate the industry, i.e., transition risk should not be dynamically adjusted. Moreover, a transition risk where the threat only forces a firm to switch from a brown to a green technology and thus pay the adjustment cost $C^{G}$ does not prevent the failed green transition as long as the brown technology is sufficiently profitable, i.e., as long as $\Pi^{BB}>\Pi^{GB}-C^{G}$. In this case, the risk of a failed green transition can be avoided either by imposing greenness through regulation or by imposing a tax on the brown technology that makes it unprofitable after deducting the cost of the green transition. In conclusion, reinforcing the transition risk by threatening to tax brown production is the only green economic policy to ensure the green transition. This is especially true when the brown technology is more profitable than the green technology.

\subsection{Effects of the adjustment costs}

With regard to the adjustment costs, the following corollary shows that they may hinder the adoption of the most profitable technology.

\medskip

\begin{corollary}\label{Corr3}
Note that:
\begin{itemize}
\item[(1)] if $\Pi^{GG}>\Pi^{GB} > \Pi^{BG}> \Pi^{BB}$, then only Scenarios 1, 6,  and 9 are possible. In this case, if the cost of adapting the green technology is subsidized, only Scenario 9 is possible; if the cost of adapting the green technology is sufficiently high, only Scenario 1 is possible;
\item[(2)] if $\Pi^{BG}> \Pi^{BB}>\Pi^{GG}>\Pi^{GB}$, then only Scenarios 1, 2, 4, and 5 are possible. In this case, a zero-cost of adjustment to the brown technology implies Scenario 5, while a sufficiently high cost of adjustment to the brown technology implies Scenario 1.
\end{itemize}
\end{corollary}

\medskip

Corollary \ref{Corr3}-(1) indicates that when the green technology is the most profitable the green transition is always possible. However, the certainty of the green transition is only achieved when the cost of adopting the green transition is subsidized. Corollary \ref{Corr3}-(2) indicates that when the brown technology is more profitable, the green transition may be a possibility. However, it depends on the cost of adopting a brown technology, which must be sufficiently high. Moreover, it states that a failed green transition is always a possibility when the brown technology is more profitable. All in all, the adjustment costs hinder the adoption of the most profitable technology, impeding the green transition when the most profitable technology is the green one, and favor the chances of a green transition when the most profitable technology is the brown one.

In conclusion, subsidizing the green technology ensures the green transition only if the green technology is more profitable than the brown one. Otherwise, in addition to subsidizing the green technology, an increase in the transition risk that makes the brown technology less profitable is also required to ensure the green transition. As already discussed, increasing transition risk can be achieved by threatening to impose higher taxes on the brown technology.

\section{Conclusions}\label{Conc}

In this paper, we propose an evolutionary competition model to study how policymakers may promote a green transition and thereby reduce the physical risk. Our setup allows for adjustment costs due to the adoption of green technologies, endogenous technological progress that generates increasing returns in green technologies, and a dynamically adjusted transition risk, i.e., a transition risk that is inversely related to the level of greenness of the industry and is the result of a threatened adaptive climate policy. The investigation of the global dynamics of the model suggests that the threat of a green tax is an appropriate policy to promote the green transition. A minimum level of such a tax is estimated, with the policy recommendation that the threat of imposing a green tax remains high regardless of the level of greenness of the industry. The effect of this environmental policy is enforced by technological progresses, which increases the returns to adopting a low-carbon technology. Finally, even if green technologies are more profitable than brown technologies, the adjustment cost of adopting low-carbon technologies may impede the green transition.

Our model setup is innovative and suitable for a wide range of applications where switching strategies entails costs. Nevertheless, the setup can be further extended to account for an endogenous variable measuring the environmental quality or to account for economies of scope, i.e., complementarities that exist between environmental activities and conventional, non-environmental production, see, e.g., \citet{BaileyFriedlaender1982}. Another behavioral aspect that is well-suited to investigate with an evolutionary approach such as ours is the so-called rebound effect: energy efficiency improvements can cause a reduction in the cost of energy services, which may lead to some increase in the demand for these services, see, e.g., \citet{Sorrell2007}. For example, improvements in the fuel efficiency of vehicles may effectively reduce the cost per mile travelled, which may lead to users to choose to travel further, thereby offsetting some of the expected emissions reductions. Similar arguments have been advanced in the case of efficiency improvements in lighting, see \citet{FouquetPearson2011}. Moreover, some have argued that there may also be macro-economic rebound effects from the widespread diffusion of efficiency improvements and technological change, as they free up economic resources that can be invested in creating new ways to meet growing end-user demand through increased output.

It is also important to note that our analysis abstracts from two additional avenues through which climate policy risk can potentially affect the green transition. First, we do not capture interactions between climate policy risk and climate damages. Second, our model abstracts from endogenous innovation. Intuitively, the possibility of a future carbon tax could increase the expected returns to innovation in clean technologies relative to fossil technologies. This innovative response to climate policy risk could further reduce the ratio of fossil to clean capital, magnifying the composition effect and the resulting emissions reductions. Thus, our results should be viewed as a lower bound on the effectiveness of climate policy transition risk in fostering the green transition.

\bigskip

\textbf{Acknowledgements:} This work was funded by the European Union - Next Generation EU, Mission 4: ``Education and Research" - Component 2: ``From research to business", through PRIN 2022 under the Italian Ministry of University and Research (MUR). Project: 2022JRY7EF - Qnt4Green - Quantitative Approaches for Green Bond Market: Risk Assessment, Agency Problems and Policy Incentives - CUP: J53D23004700008.

\bigskip

\textbf{Compliance with Ethical Standards:}
The authors declare no conflict of interest.

\appendix

\section{General properties of model \eqref{System::Main}.}\label{AppA}

Model \eqref{System::Main} is a pair of two exponential replicator dynamics adjusted to account for the costs of switching production technologies and the following results hold.

\medskip

\begin{lemma}\label{Lemma:MainProperties}
Assume $C^{B},C^{G}\geq0$, either $C^{B}>0$ or $C^{G}>0$, $a\neq 0$ and $\beta>0$. With regard to model \eqref{System::Main}, we have that:
\begin{itemize}
\item[(P1)] Box $\mathcal{D}=\left[0,1\right]^{2}$ and segment $\mathcal{A}= \left\{\left.\left(\eta^{1,G},\eta^{2,G}\right)\right|\eta^{1,G}=\eta^{2,G}\text{, }\eta^{1,G}\in\left[0,1\right]\right\}$ are invariant sets;
\item[(P2)] The model is symmetric with respect to $\mathcal{A}$;
\item[(P3)] Any invariant set (e.g., attractor) is either symmetric with respect to $\mathcal{A}$ or an invariant set symmetric to it with respect to $\mathcal{A}$ exists;
\item[(P4)] Two invariant sets symmetric with respect to $\mathcal{A}$ are both either attracting or repelling.
\item[(P5)] The four vertices of $\mathcal{D}$, $\eta_{00}=\left(0,0\right)$, $\eta_{10}=\left(1,0\right)$, $\eta_{11}=\left(1,1\right)$, and $\eta_{01}=\left(0,1\right)$ are equilibria of the model.
\item[(P6)] Excluding the vertex, at each edge of $\mathcal{D}$ there is at most one equilibrium.
\item[(P7)] In addition to the vertex, the possible equilibria at the edges of $\mathcal{D}$ are $\left(0,\eta^{*}\right)$, $\left(\eta^{*},0\right)$, $\left(1,\eta^{+}\right)$, and $\left(\eta^{+},1\right)$, with
\begin{equation}
\eta^{*} =  \resizebox{0.4\hsize}{!}{$\frac{1 - \exp\left(\beta\left(\Pi^{BB} - \Pi^{GB}+C^{G}\right)\right)}{1-  \exp\left(\beta\left(\Pi^{BB} - \Pi^{GB}+C^{G}\right)\right) + 1 -  \exp\left(\beta \left(\Pi^{GB} - \Pi^{BB} + C^{B}\right)\right)}$} \quad \text{and}\quad \eta^{+} = \resizebox{0.4\hsize}{!}{$\frac{1 - \exp\left(\beta\left(\Pi^{BG}-\Pi^{GG}+C^{G}\right)\right)}{1 - \exp\left(\beta\left(\Pi^{BG}-\Pi^{GG}+C^{G}\right)\right)+ 1 - \exp\left(\beta \left(\Pi^{GG}-\Pi^{BG}+C^{B}\right)\right)}$}\text{.}
\end{equation}
\item[(P8)] $\left(0,\eta^{*}\right)$ and $\left(\eta^{*},0\right)$ are equilibria if and only if $\Pi^{GB} + C^{B}>\Pi^{BB}> \Pi^{GB}-C^{G}$;
\item[(P9)] $\left(1,\eta^{+}\right)$ and $\left(\eta^{+},1\right)$ are equilibria if and only if $\Pi^{BG}+C^{G}>\Pi^{GG}> \Pi^{BG} - C^{B}$;
\item[(P10)] The inner equilibria are the intersection points of the following two curves in $\mathcal{D}$:
\begin{equation}\label{eqconinnereq}
\begin{array}{lll}
\eta^{1,G}  & = & \frac{1-\exp\left(\beta\left(a\eta^{2,G}+b+C^{G}\right)\right)}{2 - \exp\left(\beta\left(a\eta^{2,G}+b+C^{G}\right)\right) - \exp\left(\beta \left(-a\eta^{2,G}-b+C^{B}\right)\right)}  \\
\\
\eta^{2,G} & = & \frac{1-\exp\left(\beta\left(a\eta^{1,G}+b+C^{G}\right)\right)}{2 - \exp\left(\beta\left(a\eta^{1,G}+b+C^{G}\right)\right)- \exp\left(\beta \left(-a\eta^{1,G}-b+C^{B}\right)\right)}
\end{array}
\end{equation}
and $\left(\bar{\eta},\bar{\eta}\right)\in\mathcal{A}$, with $\bar{\eta}=0.5$, is an inner equilibrium if and only if $\Pi^{GG}-C^{G}-\Pi^{BG}=\Pi^{BB}-C^{B}-\Pi^{GB}$.
\item[(P11)] $\eta_{10}$ and $\eta_{01}$ are either both attractors, or both repellors, or both saddles.
\item[(P12)] $\eta_{00}$ and $\eta_{11}$ are either attractors or repellors, that is, they cannot be saddles.
\item[(P13)] When $\eta_{00}$ and $\eta_{11}$ are both attractors or both repellors, an inner equilibrium exists that belongs to $\mathcal{A}$.
\item[(P14)]For $\Pi^{BB} > \Pi^{GB} - C^{G}$, equilibrium $\eta_{00}$ is locally asymptotically stable. It is unstable otherwise.
\item[(P15)] For $\Pi^{GG} > \Pi^{BG}-C^{B}$, equilibrium $\eta_{11}$ is locally asymptotically stable. It is unstable otherwise.
\item[(P16)] For $\Pi^{BG}>\Pi^{GG}-C^{G}$ and $\Pi^{GB}>\Pi^{BB} -C^{B}$, equilibria $\eta_{10}$ and $\eta_{01}$ are locally asymptotically stable. They are unstable otherwise.
\item[(P17)] If they exist, equilibria $\left(0,\eta^{*}\right)$ and $\left(\eta^{*},0\right)$ are unstable: they are saddles for $\left(\Pi^{BG}-\Pi^{GG}\right)\eta^{*} + \left(\Pi^{BB} - \Pi^{GB}\right)\left(1-\eta^{*}\right)> - C^{G}$ and repellors otherwise.
\item[(P18)] If they exist, equilibria $\left(1,\eta^{+}\right)$ and $\left(\eta^{+},1\right)$ are unstable: they are saddles for $\left(\Pi^{BG}-\Pi^{GG}\right)\eta^{+} + \left(\Pi^{BB} - \Pi^{GB}\right)\left(1-\eta^{+}\right) <C^{B}$ and repellors otherwise.
\item[(P19)] If $\left(\bar{\eta},\bar{\eta}\right)$ is an inner equilibrium, its real and distinct eigenvalues are given by
\begin{equation}\label{Eigenvaluesetaeta}
\lambda^{\left(\bar{\eta},\bar{\eta}\right)}_{1,2} = \frac{\left(\bar{\eta}\right)^{2}+\bar{\eta}\left(\left(2-\bar{\eta}\right)\mp \beta a\bar{\eta}\left(1-\bar{\eta}\right) \right)\exp\left(\beta \left(a\bar{\eta}+b-C^{B}\right)\right)}{\left[\bar{\eta}+\left(1-\bar{\eta}\right)\exp\left(\beta \left(a\bar{\eta}+b-C^{B}\right)\right)\right]^{2}}  +    \frac{\left(1\mp \beta a\bar{\eta}\right)\left(1-\bar{\eta}\right)^{2}\exp\left(\beta\left(a\bar{\eta}+b+C^{G}\right)\right)  
  -\left(\bar{\eta}\right)^{2}}{\left[\bar{\eta}+\left(1-\bar{\eta}\right)\exp\left(\beta\left(a\bar{\eta}+b+C^{G}\right)\right)\right]^{2}} \text{.}
\end{equation}
\item[(P20)] If equilibrium $\left(\bar{\eta},\bar{\eta}\right)$ exists and is stable, it loses stability through either a fold bifurcation or a  transcritical bifurcation. This bifurcation takes place along the manifold given by the invariant set $\mathcal{A}$ when $a<0$, or along the manifold transverse to the invariant set $\mathcal{A}$ when $a>0$.
\end{itemize}
\end{lemma}

\medskip

The results in Lemma \ref{Lemma:MainProperties} characterize the equilibria and the dynamics at the edges of the square box $\mathcal{D} = \left[0,1\right]^{2}$. The existence of an inner equilibrium is also proven in certain cases, that is, when $\eta_{00}$ and $\eta_{11}$ are both either stable or unstable. Numerical simulations have shown that this inner equilibrium can be unstable, see Figure \ref{Figure19}. However, it remains an open question whether an inner equilibrium can be stable. Instead, the existence of several unstable inner equilibria is confirmed by the numerical simulations of Figure \ref{FigureCasoExtra}. Specifically, Figure \ref{FigureCasoExtra}-(a) shows an example of Scenario 1 of Proposition \ref{Prop:FullModelSce}, where there are three unstable inner equilibria, one of which is on the invariant diagonal $\mathcal{A}$, while the other two are outside this diagonal and are symmetric to each other with respect to this diagonal. Despite the presence of multiple unstable inner equilibria, Figure \ref{FigureCasoExtra}-(a) reports an example of the dynamics of Scenario 1 that is qualitatively equivalent to the one of Figure \ref{Figure19}-(a) in the sense that the only alternatives are a green transition, a partial green transition, or a failed green transition. Another example of three unstable inner equilibria, but along the diagonal $\mathcal{A}$, can be seen in Figure \ref{FigureCasoExtra}-(b). We are again in the case of a Scenario 1 and, except for the three inner equilibria on the diagonal, the dynamics is still qualitatively equivalent to the one of Figure \ref{Figure19}-(a). A numerical example with an inner unstable equilibrium and a saddle 2-cycle on the diagonal $\mathcal{A}$ is reported in Figure \ref{FigureCasoExtra}-(c). Except for the presence of the saddle 2-cycle, this numerical example confirms the global dynamics of Scenario 8 shown in Figure \ref{Figure19}-(h). An example of two inner equilibria in the diagonal $\mathcal{A}$ is reported in Figure \ref{FigureCasoExtra}-(d). Even in this case, except for the presence of the two inner equilibria, the dynamics is that of Scenario 4 described and shown in Figure \ref{Figure19}-(d).

These numerical examples emphasize that, despite the presence of multiple inner equilibria, the description of the possible outcomes of the model is consistent with those outcomes observed and discussed in Section \ref{fullfledgedmodel}. An extensive numerical investigation, which is omitted for reasons of space, confirms this, and hence the robustness of the results in Section \ref{fullfledgedmodel}.

\begin{figure}[!t]
    \begin{center}
    \includegraphics[scale=0.75]{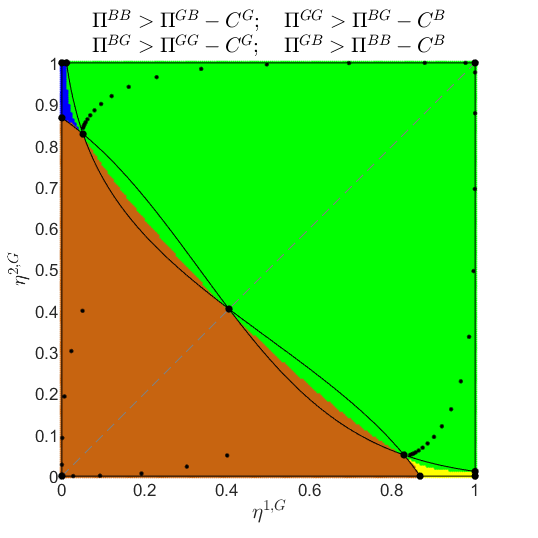}
    \includegraphics[scale=0.75]{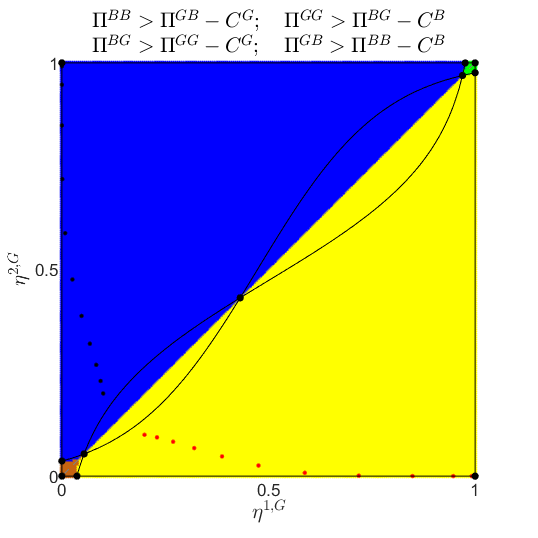}\\
    \includegraphics[scale=0.75]{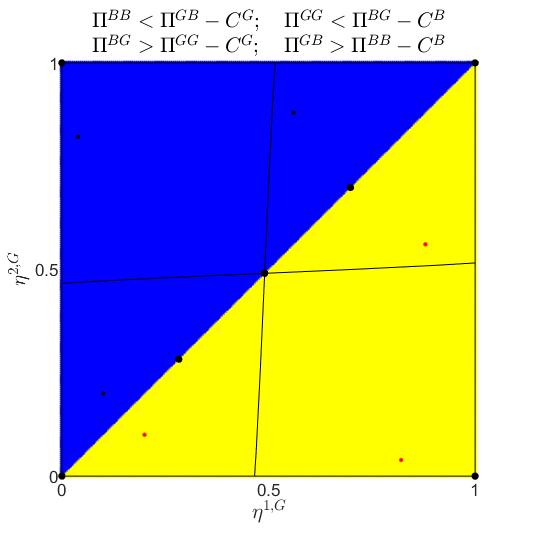}
    \includegraphics[scale=0.75]{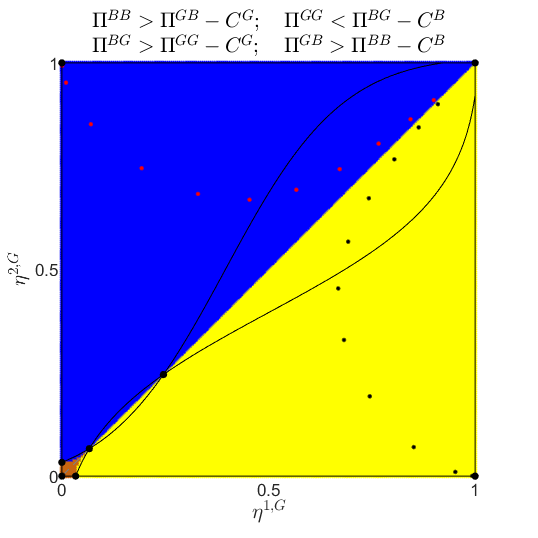}\\
        \caption{\small Left-top panel, from Scenario 1 an example of multiple inner and unstable equilibria: $\beta = 5$; 
$\Pi^{BB} = 1.9$; $\Pi^{GB} = 1.7$; $C^{G} = 0.3$; $\Pi^{GG} = 2.75$; $\Pi^{BG} = 2.5$; $C^{B} = 0.4$; $\beta = 5$. Right-top panel, from Scenario 1 an example of three inner equilibria along the diagonal, two repellors and a saddle: $\Pi^{BB} = 1$; $\Pi^{GB} = 1.95$; $C^{G} = 1.2$; $\Pi^{GG} = 5.3$; $\Pi^{BG} = 5.1$; $C^{B} = 0.01$; $\beta = 4$. Left-bottom panel, from Scenario 8 an example of a repellor inner equilibrium and a saddle 2-cycle along the diagonal: $\Pi^{BB} = 0.8$; $\Pi^{GB} = 1.95$; $C^{G} = 0.1$; $\Pi^{GG} = 4$; $\Pi^{BG} = 5.1$; $C^{B} = 0.01$; $\beta = 4$. Right-bottom panel, from Scenario 4 an example of two inner equilibria along the diagonal, a repellor and a saddle: $\Pi^{BB} = 0.9$; $\Pi^{GB} = 1.95$; $C^{G} = 1.2$; $\Pi^{GG} = 2$; $\Pi^{BG} = 2.1$; $C^{B} = 0.01$.}\label{FigureCasoExtra}
    \end{center}
\end{figure}

\section{Technical results and proofs}\label{AppB}

\begin{proof}[Proof of Proposition \ref{Prop::ExpReplicatorStandard}]
The model \eqref{MAINMAPNested2} can be rewritten as
\begin{equation}
\eta^{i,G}_{t+1} = \frac{\eta^{i,G}_{t}}{\eta^{i,G}_{t} +\left(1-\eta^{i,G}_{t}\right)\exp\left(\beta \left(\Pi^{B}-\Pi^{G}\right)\right)}\text{.}  
\end{equation}
It follows that $\eta^{i,G}_{t+1}\in\left[0,1\right]$ for $\eta^{i,G}_{t}\in\left[0,1\right]$, that is, $\left[0,1\right]$ is invariant for the model. Moreover, note that for each $\eta^{i,G}_{t}\in\left(0,1\right)$, we have that
\begin{equation}
\frac{\eta^{i,G}_{t}}{\eta^{i,G}_{t} +\left(1-\eta^{i,G}_{t}\right)\exp\left(\beta \left(\Pi^{B}-\Pi^{G}\right)\right)}>\eta^{i,G}_{t} 
\end{equation}
if and only if $\Pi^{B}-\Pi^{G}<0$. Therefore, $\eta^{i,G}_{t+1}>\eta^{i,G}_{t}$ for all $\eta^{i,G}_{t}\in\left(0,1\right)$. Since $\left[0,1\right]$ is an invariant set, Proposition \ref{Prop::ExpReplicatorStandard}-(1) follows. In a similar way, we can prove Proposition \ref{Prop::ExpReplicatorStandard}-(2). For $\Pi_{G}=\Pi_{B}$, note that $\eta^{i,G}_{t+1}=\eta^{i,G}_{t}$ for all $\eta^{i,G}_{t}$. This proves Proposition \ref{Prop::ExpReplicatorStandard}-(3).
\end{proof}

\medskip

\begin{proof}[Proof of Proposition \ref{Prop::ExpReplicatorEvoluta}]
Imposing $\bar{\eta}^{i,G}_{t+1}=\bar{\eta}^{i,G}_{t}=\bar{\eta}^{i,G}$ in model \eqref{MAINMAPNested1}, we obtain the equilibrium equation
\begin{equation}
\bar{\eta}^{i,G}  =  \bar{\eta}^{i,G}\frac{\bar{\eta}^{i,G} \exp\left(\beta \Pi^{G}\right)}{\bar{\eta}^{i,G} \exp\left(\beta \Pi^{G}\right)+\left(1-\bar{\eta}^{i,G}\right)\exp\left(\beta \left(\Pi^{B}-C^{B}\right)\right)} + \left(1-\bar{\eta}^{i,G}\right) \frac{\bar{\eta}^{i,G} \exp\left(\beta\left(\Pi^{G}-C^{G}\right)\right)}{\bar{\eta}^{i,G} \exp\left(\beta \left(\Pi^{G}-C^{G}\right)\right)+\left(1-\bar{\eta}^{i,G}\right)\exp\left(\beta \Pi^{B}\right)}\text{,} 
\end{equation}
which can be rewritten as
\begin{equation}\label{EqEquation}
\frac{\bar{\eta}^{i,G}\left(1-\bar{\eta}^{i,G}\right)\exp\left(\beta \left(\Pi^{B}-C^{B}\right)\right)}{\bar{\eta}^{i,G} \exp\left(\beta \Pi^{G}\right)+\left(1-\bar{\eta}^{i,G}\right)\exp\left(\beta \left(\Pi^{B}-C^{B}\right)\right)}  =   \frac{\bar{\eta}^{i,G}\left(1-\bar{\eta}^{i,G}\right) \exp\left(\beta\left(\Pi^{G}-C^{G}\right)\right)}{\bar{\eta}^{i,G} \exp\left(\beta \left(\Pi^{G}-C^{G}\right)\right)+\left(1-\bar{\eta}^{i,G}\right)\exp\left(\beta \Pi^{B}\right)}\text{.}  
\end{equation}
It is straightforward to verify that $\bar{\eta}^{i,G}=0$ and $\bar{\eta}^{i,G}=1$ are always solutions and therefore equilibria of the model. For $\bar{\eta}^{i,G}\in\left(0,1\right)$, the equilibrium equation \eqref{EqEquation} is equivalent to
\begin{equation}\label{inneqcond1d}
\bar{\eta}^{i,G}+\left(1-\bar{\eta}^{i,G}\right)\exp\left(\beta \left(\Pi^{B} - \Pi^{G} + C^{G}\right)\right) = \bar{\eta}^{i,G} \exp\left(\beta \left(\Pi^{G}-\Pi^{B}+C^{B}\right)\right)+\left(1-\bar{\eta}^{i,G}\right)\text{,} 
\end{equation}
which is an identity for $\Pi^{B} = \Pi^{G} - C^{G} = \Pi^{G}+C^{B}$ and a linear equation otherwise. However, $\Pi^{B} = \Pi^{G} - C^{G} = \Pi^{G}+C^{B}$ is excluded by the assumptions $C^{B},C^{G}\geq0$ and either $C^{B}>0$ or $C^{G}>0$. Therefore, \eqref{inneqcond1d} has a unique solution given by
\begin{equation}\label{innereta}
\bar{\eta}^{i,G}_{in} = \frac{\exp\left(\beta\left(\Pi_{B}-\Pi_{G}+C^{G}\right)\right)-1}{\exp\left(\beta\left(\Pi_{B}-\Pi_{G}+C^{G}\right)\right)-1+\exp\left(\beta\left(\Pi_{G}-\Pi_{B}+C^{B}\right)\right)-1}\text{.} 
\end{equation}
This solution is the unique inner equilibrium of model \eqref{MAINMAPNested1} when $\bar{\eta}^{i,G}_{in}\in\left(0,1\right)$ and no inner equilibria exist otherwise. Let us discuss when $\bar{\eta}^{i,G}_{in}\in\left(0,1\right)$. There are four possible cases:
\begin{itemize}
\item[] Case 1: $\Pi_{B}-\Pi_{G}+C^{G}>0$ and $\Pi_{G}-\Pi_{B}+C^{B}>0$ (or equivalently $\Pi_{G}+C^{B}>\Pi_{B}>\Pi_{G}-C^{G}$). Hence, $0<\exp\left(\beta\left(\Pi_{B}-\Pi_{G}+C^{G}\right)\right)-1<\exp\left(\beta\left(\Pi_{B}-\Pi_{G}+C^{G}\right)\right)+\exp\left(\beta\left(\Pi_{G}-\Pi_{B}+C^{B}\right)\right)-2$. Then, we have that $\bar{\eta}^{i,G}\in\left(0,1\right)$.
\item[] Case 2: $\Pi_{B}-\Pi_{G}+C^{G}>0>\Pi_{G}-\Pi_{B}+C^{B}$. Hence, $\exp\left(\beta\left(\Pi_{B}-\Pi_{G}+C^{G}\right)\right)-1>0$ and $\exp\left(\beta\left(\Pi_{B}-\Pi_{G}+C^{G}\right)\right)-1>\exp\left(\beta\left(\Pi_{B}-\Pi_{G}+C^{G}\right)\right)+\exp\left(\beta\left(\Pi_{G}-\Pi_{B}+C^{B}\right)\right)-2$. Then, imposing $\bar{\eta}^{i,G}>0$ implies having $\exp\left(\beta\left(\Pi_{B}-\Pi_{G}+C^{G}\right)\right)-1>\exp\left(\beta\left(\Pi_{B}-\Pi_{G}+C^{G}\right)\right)+\exp\left(\beta\left(\Pi_{G}-\Pi_{B}+C^{B}\right)\right)-2>0$. Hence, $\bar{\eta}^{i,G}>1$ and an inner equilibrium cannot exist.
\item[] Case 3: $\Pi_{B}-\Pi_{G}+C^{G}<0<\Pi_{G}-\Pi_{B}+C^{B}$. Hence, $\exp\left(\beta\left(\Pi_{B}-\Pi_{G}+C^{G}\right)\right)-1<0$ and $\exp\left(\beta\left(\Pi_{B}-\Pi_{G}+C^{G}\right)\right)-1<\exp\left(\beta\left(\Pi_{B}-\Pi_{G}+C^{G}\right)\right)+\exp\left(\beta\left(\Pi_{G}-\Pi_{B}+C^{B}\right)\right)-2$. Then, $\bar{\eta}^{i,G}>0$ implies $\exp\left(\beta\left(\Pi_{B}-\Pi_{G}+C^{G}\right)\right)-1<\exp\left(\beta\left(\Pi_{B}-\Pi_{G}+C^{G}\right)\right)+\exp\left(\beta\left(\Pi_{G}-\Pi_{B}+C^{B}\right)\right)-2<0$. Hence, $\bar{\eta}^{i,G}>1$ and an inner equilibrium cannot exist.
\item[] Case 4: $\Pi_{B}-\Pi_{G}+C^{G}<0$ and $\Pi_{G}-\Pi_{B}+C^{B}<0$, or equivalently $\Pi_{G}+C^{B}<\Pi_{B}<\Pi_{G}-C^{G}$, which is not possible since $C^{G},C^{B}\geq0$, and either $C^{G}>0$ or $C^{B}>0$, by assumption.
\end{itemize}
It follows that the inner equilibrium exists if and only if $\Pi_{G}+C^{B}>\Pi_{B}>\Pi_{G}-C^{G}$.
Regarding the stability of the equilibria, we have that
\begin{equation}
\begin{array}{lll}
\frac{\partial \eta^{i,G}_{t+1} }{\partial \eta^{i,G}_{t}} & = &   \frac{\eta^{i,G}_{t} \exp\left(\beta \Pi^{G}\right)}{\eta^{i,G}_{t} \exp\left(\beta \Pi^{G}\right)+\left(1-\eta^{i,G}_{t}\right)\exp\left(\beta \left(\Pi^{B}-C^{B}\right)\right)}\\
\\
&& + \frac{\eta^{i,G}_{t}\left(\exp\left(\beta \Pi^{G}\right)\left[\eta^{i,G}_{t} \exp\left(\beta \Pi^{G}\right)+\left(1-\eta^{i,G}_{t}\right)\exp\left(\beta \left(\Pi^{B}-C^{B}\right)\right)\right] - \eta^{i,G}_{t} \exp\left(\beta \Pi^{G}\right)\left[\exp\left(\beta \Pi^{G}\right)-\exp\left(\beta \left(\Pi^{B}-C^{B}\right)\right)\right]\right)}{\left[\eta^{i,G}_{t}\exp\left(\beta \Pi^{G}\right)+\left(1-\eta^{i,G}_{t}\right)\exp\left(\beta \left(\Pi^{B}-C^{B}\right)\right)\right]^{2}}  \\
\\
&& - \frac{\eta^{i,G}_{t} \exp\left(\beta\left(\Pi^{G}-C^{G}\right)\right)}{\eta^{i,G}_{t} \exp\left(\beta \left(\Pi^{G}-C^{G}\right)\right)+\left(1-\eta^{i,G}_{t}\right)\exp\left(\beta \Pi^{B}\right)} \\
\\
&& +  \frac{\left(1-\eta^{i,G}_{t}\right)\left(\exp\left(\beta\left(\Pi^{G}-C^{G}\right)\right)\left[\eta^{i,G}_{t} \exp\left(\beta \left(\Pi^{G}-C^{G}\right)\right)+\left(1-\eta^{i,G}_{t}\right)\exp\left(\beta \Pi^{B}\right)\right] - 
\eta^{i,G}_{t} \exp\left(\beta\left(\Pi^{G}-C^{G}\right)\right)\left[\exp\left(\beta \left(\Pi^{G}-C^{G}\right)\right)-\exp\left(\beta \Pi^{B}\right)\right]\right)}{\left[\eta^{i,G}_{t} \exp\left(\beta \left(\Pi^{G}-C^{G}\right)\right)+\left(1-\eta^{i,G}_{t}\right)\exp\left(\beta \Pi^{B}\right)\right]^{2}}\text{.}
\end{array}
\end{equation}
Therefore,
\begin{equation}
\begin{array}{lll}
\left. \frac{\partial \eta^{i,G}_{t+1} }{\partial \eta^{i,G}_{t}}\right|_{0} & = &   \exp\left(\beta\left(\Pi^{G}-C^{G}-\Pi^{B}\right)\right)
\end{array}
\end{equation}
and by standard eigenvalue analysis, it follows that the equilibrium is locally asymptotically stable if and only if $\Pi^{B}>\Pi^{G}-C^{G}$. Moreover,
\begin{equation}
\left. \frac{\partial \eta^{i,G}_{t+1} }{\partial \eta^{i,G}_{t}}\right|_{1} =  \exp\left(\beta \left(\Pi^{B}-C^{B}-\Pi^{G}\right)\right)
\end{equation}
and by standard eigenvalue analysis we can conclude that the equilibrium is locally asymptotically stable if and only if $\Pi^{B}<\Pi^{G} + C^{B}$. We have proved that the inner equilibrium exists if and only if $\bar{\eta}^{i,G}_{0}$ and $\bar{\eta}^{i,G}_{1}$ are both locally asymptotically stables, that is, if and only if $\Pi_{G}+C^{B}>\Pi_{B}>\Pi_{G}-C^{G}$. Since the model is smooth and there is a unique equilibrium for $\Pi_{G}+C^{B}>\Pi_{B}>\Pi_{G}-C^{G}$, we have that stability of $\bar{\eta}^{i,G}_{0}$ implies $\bar{\eta}^{i,G}_{t}>\bar{\eta}^{i,G}_{t+1}$ for all $\bar{\eta}^{i,G}_{t}\in\left(0,\bar{\eta}^{i,G}_{in}\right)$ and stability of $\bar{\eta}^{i,G}_{1}$ implies $\bar{\eta}^{i,G}_{t+1}>\bar{\eta}^{i,G}_{t}$ for all $\bar{\eta}^{i,G}_{t}\in\left(\bar{\eta}^{i,G}_{in},1\right)$. Therefore, a monotonic convergence to $\bar{\eta}^{i,G}_{0}$ for any trajectory starting in $\mathcal{B}\left(\bar{\eta}_{0}^{i,G}\right)=\left[0,\bar{\eta}^{i,G}_{in}\right)$ and a monotonic convergence to $\bar{\eta}^{i,G}_{0}$ for any trajectory starting in $\mathcal{B}\left(\bar{\eta}_{1}^{i,G}\right)=\left(\bar{\eta}^{i,G}_{in},1\right]$ occurs. Let us further note that for $\Pi_{G}+C^{B}=\Pi_{B}$, the unique potential inner equilibrium $\bar{\eta}^{i,G}_{in}$ defined in \eqref{innereta} merges with equilibrium $\bar{\eta}_{1}^{i,G}=1$ and their eigenvalues are equal to $1$. In this case, we obtain
\begin{equation}
\eta^{i,G}_{t+1}=\eta^{i,G}_{t} \left[\eta^{i,G}_{t}  + \left(1-\eta^{i,G}_{t}\right) \frac{1}{\eta^{i,G}_{t}  +\left(1-\eta^{i,G}_{t}\right)\exp\left(\beta \left(\Pi^{B} - \Pi^{G} 
 + C^{G}\right)\right)}\right]\text{.}
\end{equation}
Let us further note that $\Pi_{G}+C^{B}=\Pi_{B}$ implies $\Pi_{B}>\Pi_{G}-C^{G}$. Therefore, $ \eta^{i,G}_{t} > \eta^{i,G}_{t+1}$. This proves $\mathcal{B}\left(\bar{\eta}_{0}^{i,G}\right)=\left[0,1\right)$ when $\Pi_{G}+C^{B}=\Pi_{B}$. Similarly, we can prove that $\bar{\eta}^{i,G}_{in}$ merges with $\bar{\eta}_{0}^{i,G}=0$ and  $\mathcal{B}\left(\bar{\eta}_{1}^{i,G}\right)=\left(0,1\right]$ when $\Pi_{B}=\Pi_{G}-C^{G}$. For $\Pi_{B}>\Pi_{G}+C^{B}$, we have proven that $\bar{\eta}_{0}^{i,G}$ is asymptotically stable and there are no equilibria in $\left(0,1\right)$. Asymptotic stability of $\bar{\eta}_{0}^{i,G}$ implies $\eta^{i,G}_{t+1}<\eta^{i,G}_{t}$ in a neighborhood of $\bar{\eta}_{0}^{i,G}$. Since model \eqref{MAINMAPNested1} is continuous, no equilibria in $\left(0,1\right)$ and $\eta^{i,G}_{t+1}<\eta^{i,G}_{t}$ in a neighborhood of $\bar{\eta}_{0}^{i,G}$ imply $\eta^{i,G}_{t+1}<\eta^{i,G}_{t}$ in $\left(0,1\right)$. Hence, $\mathcal{B}\left(\bar{\eta}_{0}^{i,G}\right)=\left[0,1\right)$. With similar arguments, we can prove $\mathcal{B}\left(\bar{\eta}_{1}^{i,G}\right)=\left(0,1\right]$ when $\Pi_{G}-C^{G}>\Pi_{B}$.
\end{proof}

\medskip

\begin{proof}[Proof of Corollary \ref{Corr1}]
Standard calculus reveals that
\begin{equation}
\begin{array}{lll}
\lim_{\beta\rightarrow 0^{+}}  \bar{\eta}^{i,G}_{in}  & = &     \lim_{\beta\rightarrow 0^{+}}   \frac{\exp\left(\beta\left(\Pi_{B}-\Pi_{G}+C^{G}\right)\right)-1}{\exp\left(\beta\left(\Pi_{B}-\Pi_{G}+C^{G}\right)\right)-1+\exp\left(\beta\left(\Pi_{G}-\Pi_{B}+C^{B}\right)\right)-1}\\
\\
& = & \lim_{\beta\rightarrow 0^{+}}   \frac{\beta\left(\Pi_{B}-\Pi_{G}+C^{G}\right)}{\beta\left(\Pi_{B}-\Pi_{G}+C^{G}\right)+\beta\left(\Pi_{G}-\Pi_{B}+C^{B}\right)}\\
\\
& = & \lim_{\beta\rightarrow 0^{+}}   \frac{ \Pi_{B}-\Pi_{G}+C^{G} }{C^{G}+C^{B}} = \frac{ \Pi_{B}-\Pi_{G}+C^{G} }{C^{G}+C^{B}} = \bar{\eta}^{i,G}_{in,\left(\beta\rightarrow 0\right)}\text{.}\\
\end{array}
\end{equation}
Assumption $\Pi^{B}-\Pi^{G}+C^{G}>0$ yields $\bar{\eta}^{i,G}_{in,\left(\beta\rightarrow 0\right)}>0$, while assumption $C^{B}>\Pi^{B}-\Pi^{G}$ implies $\bar{\eta}^{i,G}_{in,\left(\beta\rightarrow 0\right)}<1$. By Proposition \ref{Prop::ExpReplicatorEvoluta}, we have $\mathcal{B}\left(\bar{\eta}_{0}^{i,G}\right)=\left[0,\bar{\eta}^{i,G}_{in}\right)$. Therefore $\mathcal{B}\left(\bar{\eta}_{0}^{i,G}\right)=\left[0,\bar{\eta}^{i,G}_{in,\left(\beta\rightarrow 0\right)}\right)$ when $\beta\rightarrow 0$. Note that $\frac{\partial \bar{\eta}^{i,G}_{in,\left(\beta\rightarrow 0\right)}}{\partial \left(\Pi^{B}-\Pi^{G}\right)}>0$. This proves (1). From standard calculus, we can conclude that
\begin{equation}
\begin{array}{lll}
\lim_{\beta\rightarrow +\infty}  \bar{\eta}^{i,G}_{in}  & = &     \lim_{\beta\rightarrow +\infty}   \frac{\exp\left(\beta\left(\Pi_{B}-\Pi_{G}+C^{G}\right)\right)-1}{\exp\left(\beta\left(\Pi_{B}-\Pi_{G}+C^{G}\right)\right)-1+\exp\left(\beta\left(\Pi_{G}-\Pi_{B}+C^{B}\right)\right)-1}\\
\\
& = & \lim_{\beta\rightarrow +\infty}   \frac{1-\exp\left(-\beta\left(\Pi_{B}-\Pi_{G}+C^{G}\right)\right)}{1-\exp\left(-\beta\left(\Pi_{B}-\Pi_{G}+C^{G}\right)\right)+\exp\left(\beta\left(2\Pi_{G}-2\Pi_{B}+C^{B}-C^{G}\right)\right)-\exp\left(-\beta\left(\Pi_{B}-\Pi_{G}+C^{G}\right)\right)}\\
\\
& = & \lim_{\beta\rightarrow +\infty}    \frac{1}{1+\exp\left(\beta\left(2\Pi_{G}-2\Pi_{B}+C^{B}-C^{G}\right)\right)}\text{,}\\
\end{array}
\end{equation}
which proves \eqref{betainfty}. Moreover, by Proposition \ref{Prop::ExpReplicatorEvoluta}, we have $\mathcal{B}\left(\bar{\eta}_{0}^{i,G}\right)=\left[0,\bar{\eta}^{i,G}_{in}\right)$, which proves (2). Finally,
\begin{equation}
\begin{array}{lll}
\frac{\partial \bar{\eta}^{i,G}_{in}}{\partial \beta} & = &  \frac{\left(\Pi^{B}-\Pi^{G}+C^{G}\right)\exp\left(\beta\left(\Pi^{B}-\Pi^{G}+C^{G}\right)\right)\left[\exp\left(\beta\left(\Pi^{G}-\Pi^{B}+C^{B}\right)\right)-1\right] - \left(\Pi^{G}-\Pi^{B}+C^{B}\right)\exp\left(\beta\left(\Pi^{G}-\Pi^{B}+C^{B}\right)\right)\left[\exp\left(\beta\left(\Pi^{B}-\Pi^{G}+C^{G}\right)\right)-1\right]}{\left[\exp\left(\beta\left(\Pi^{B}-\Pi^{G}+C^{G}\right)\right)-1+\exp\left(\beta\left(\Pi^{G}-\Pi^{B}+C^{B}\right)\right)-1\right]^{2}}\\
\\
& = & \frac{\left(\Pi^{B}-\Pi^{G}+C^{G} -\Pi^{G}+\Pi^{B}-C^{B}\right)\exp\left(\beta\left(C^{B}+C^{G}\right)\right)+\left(\Pi^{B}-\Pi^{G}+C^{G} -\left(\Pi^{G}-\Pi^{B}+C^{B}\right)\exp\left(\beta\left(2\left(\Pi^{G}-\Pi^{B}\right) + C^{B}-C^{G}\right)\right)\right)\exp\left(\beta\left(\Pi^{B}-\Pi^{G}+C^{G}\right)\right)}{\left[\exp\left(\beta\left(\Pi^{B}-\Pi^{G}+C^{G}\right)\right)-1+\exp\left(\beta\left(\Pi^{G}-\Pi^{B}+C^{B}\right)\right)-1\right]^{2}}\text{,}
\end{array}
\end{equation}
which is positive for $2\left(\Pi^{G}-\Pi^{B}\right) + C^{B}-C^{G}<0$ and negative otherwise. This proves (3).
\end{proof}

\medskip

\begin{proof}[Proof of Lemma \ref{Lemma:MainProperties}]
By definition of model \eqref{System::Main}, we have that $\eta^{1,G}_{t},\eta^{2,G}_{t}\in\left[0,1\right]$ implies $\eta^{1,G}_{t+1},\eta^{2,G}_{t+1}\in\left[0,1\right]$ for all $t$. Moreover, $\eta^{1,G}_{t}=\eta^{2,G}_{t}$ implies $\eta^{1,G}_{t+1}=\eta^{2,G}_{t+1}$ for all $t$. This proves (P1). Note that for any $\left(\eta^{1,G}_{t},\eta^{2,G}_{t}\right)$ mapped into $\left(\eta^{1,G}_{t+1},\eta^{2,G}_{t+1}\right)$, we have $\left(\eta^{2,G}_{t},\eta^{1,G}_{t}\right)$ mapped into $\left(\eta^{2,G}_{t+1},\eta^{1,G}_{t+1}\right)$. This proves (P2). (P3) and (P4) follow from (P2). With regard to the equilibria of the system, let us impose $\eta^{1,G}_{t+1}=\eta^{1,G}_{t}=\bar{\eta}^{1,G}$ and $\eta^{2,G}_{t+1}=\eta^{2,G}_{t}=\bar{\eta}^{2,G}$ in the dynamical system \eqref{System::Main}. Then, we obtain the following algebraic system:
\begin{equation}\label{System::Equilibria}
\begin{array}{lll}
\bar{\eta}^{1,G} & = &  \bar{\eta}^{1,G}\frac{\bar{\eta}^{1,G}}{\bar{\eta}^{1,G}+\left(1-\bar{\eta}^{1,G}\right)\exp\left(\beta \left(a\bar{\eta}^{2,G}+b-C^{B}\right)\right)} + \left(1-\bar{\eta}^{1,G}\right) \frac{\bar{\eta}^{1,G}}{\bar{\eta}^{1,G}+\left(1-\bar{\eta}^{1,G}\right)\exp\left(\beta\left(a\bar{\eta}^{2,G}+b+C^{G}\right)\right)}  \\
\\
\bar{\eta}^{2,G} & = &  \bar{\eta}^{2,G}\frac{\bar{\eta}^{2,G}}{\bar{\eta}^{2,G}+\left(1-\bar{\eta}^{2,G}\right)\exp\left(\beta \left(a\bar{\eta}^{1,G}+b-C^{B}\right)\right)} + \left(1-\bar{\eta}^{2,G}\right) \frac{\bar{\eta}^{2,G}}{\bar{\eta}^{2,G}+\left(1-\bar{\eta}^{2,G}\right)\exp\left(\beta\left(a\bar{\eta}^{1,G}+b+C^{G}\right)\right)}\text{,}
\end{array}
\end{equation}
the solutions of which are equilibria of the model \eqref{System::Main}. It is straightforward to verify that $\eta_{00}$, $\eta_{10}$, $\eta_{11}$, and $\eta_{01}$ solve \eqref{System::Equilibria} and are therefore equilibria of the model. This proves (P5). (P6) follows from Proposition \ref{Prop::ExpReplicatorEvoluta}, by noting that at each edge of the box, the model reduces to a one-dimensional map equivalent to \eqref{MAINMAPNested1}. Specifically, imposing $\bar{\eta}^{1,G}=0$ and $\bar{\eta}^{2,G}\in\left(0,1\right)$, the first equation of the algebraic systems \eqref{System::Equilibria} is satisfied, while the second one reduces to the following linear algebraic equation
\begin{equation}\label{Eqexetastar}
\bar{\eta}^{2,G}+\left(1-\bar{\eta}^{2,G}\right)\exp\left(\beta\left(b+C^{G}\right)\right) =  \bar{\eta}^{2,G}\exp\left(-\beta \left(b-C^{B}\right)\right)+\left(1-\bar{\eta}^{2,G}\right)\text{,}
\end{equation}
which is an identity for $b = - C^{G} = C^{B}$ and a linear equation otherwise. However, $b = - C^{G} = C^{B}$ is excluded by the assumptions $C^{B},C^{G}\geq0$ and either $C^{B}>0$ or $C^{G}>0$. Therefore, \eqref{Eqexetastar} has a unique solution. By straightforward algebra, this unique solution is
\begin{equation}
\eta^{*} =  \frac{1 - \exp\left(\beta\left(b+C^{G}\right)\right)}{2-  \exp\left(\beta\left(b+C^{G}\right)\right) -  \exp\left(-\beta \left(b-C^{B}\right)\right)}\text{.}
\end{equation}
Hence, $\left(0,\eta^{*}\right)$ is the unique equilibrium (excluding the vertices) at edge $\eta^{1,G}=0$ of box $\left[0,1\right]^{2}$.
By symmetric arguments (see (P2) and (P3)), $\left(\eta^{*},0\right)$ is the unique equilibrium (excluding the vertices) at edge $\eta^{2,G}=0$ of box $\left[0,1\right]^{2}$. Moreover, imposing $\bar{\eta}^{1,G}=1$ and $\bar{\eta}^{2,G}\in\left(0,1\right)$, the first equation of the algebraic system \eqref{System::Equilibria} is satisfied, while the second one reduces to the following linear algebraic equation
\begin{equation}\label{Eqexetaplus} 
\bar{\eta}^{2,G}+\left(1-\bar{\eta}^{2,G}\right)\exp\left(\beta\left(a+b+C^{G}\right)\right) = \bar{\eta}^{2,G}\exp\left(-\beta \left(a+b-C^{B}\right)\right)+\left(1-\bar{\eta}^{2,G}\right)\text{,}
\end{equation}
which is an identity for $a + b = - C^{G} = C^{B}$ and a linear equation otherwise. However, $a + b = - C^{G} =  C^{B}$ is excluded by the assumptions $C^{B},C^{G}\geq0$ and either $C^{B}>0$ or $C^{G}>0$. Therefore, \eqref{Eqexetaplus} has a unique solution. By straightforward algebra, this unique solution is
\begin{equation} 
\eta^{+} = \frac{1 - \exp\left(\beta\left(a+b+C^{G}\right)\right)}{2 - \exp\left(\beta\left(a+b+C^{G}\right)\right) - \exp\left(-\beta \left(a+b-C^{B}\right)\right)}\text{.}
\end{equation}
Hence, $\left(1,\eta^{+}\right)$ is the unique equilibrium (excluding the vertices) at edge $\eta^{1,G}=1$ of box $\left[0,1\right]^{2}$. By symmetric arguments (see (P2) and (P3)), $\left(\eta^{+},1\right)$ is the unique equilibrium (excluding the vertices) at edge $\eta^{2,G}=1$ of box $\left[0,1\right]^{2}$. This proves (P7). Imposing $\eta^{*}$ and $\eta^{+}$ in $\left(0,1\right)^{2}$, (P8) and (P9) follow. Indeed, assuming $\eta^{+}>0$ and $2 - \exp\left(\beta\left(a+b+C^{G}\right)\right) - \exp\left(-\beta \left(a+b-C^{B}\right)\right)>0$ implies $a+b+C^{G}<0$, which implies $a+b-C^{B}<0$. Then, $\eta^{+}<1$ implies $1-\exp\left(-\beta \left(a+b-C^{B}\right)\right)>0$ and thus $a+b-C^{B}<0$. Instead, $\eta^{+}>0$ and $2 - \exp\left(\beta\left(a+b+C^{G}\right)\right) - \exp\left(-\beta \left(a+b-C^{B}\right)\right)<0$ implies $a+b+C^{G}>0$. Then, $\eta^{+}<1$ implies $1-\exp\left(-\beta \left(a+b-C^{B}\right)\right)<0$, which implies $a+b-C^{B}>0$. Therefore, $a+b-C^{B}>0$ and $a+b+C^{G}>0$, which are the conditions in (P8). Similarly, we obtain the conditions in (P9). Searching for solutions $\left(\bar{\eta}^{1,G},\bar{\eta}^{2,G}\right)\in\left(0,1\right)^{2}$, the algebraic system \eqref{System::Equilibria} is equivalent to:
\begin{equation}
\begin{array}{lll}
\bar{\eta}^{1,G}  & = & \frac{1-\exp\left(\beta\left(a\bar{\eta}^{2,G}+b+C^{G}\right)\right)}{2 - \exp\left(\beta\left(a\bar{\eta}^{2,G}+b+C^{G}\right)\right) - \exp\left(\beta \left(-a\bar{\eta}^{2,G}-b+C^{B}\right)\right)}  \\
\\
\bar{\eta}^{2,G} & = & \frac{1-\exp\left(\beta\left(a\bar{\eta}^{1,G}+b+C^{G}\right)\right)}{2 - \exp\left(\beta\left(a\bar{\eta}^{1,G}+b+C^{G}\right)\right) - \exp\left(\beta \left(-a\bar{\eta}^{1,G}-b+C^{B}\right)\right)}\text{.}
\end{array}
\end{equation}
By definition of equilibrium, the intersection points of these two curves in box $\left(0,1\right)^{2}$ are equilibria of the dynamical system. Imposing $\eta^{1,G} = \eta^{2,G} = \bar{\eta}=0.5$ in \eqref{eqconinnereq}, the condition $\Pi^{GG}-C^{G}-\Pi^{BG}=\Pi^{BB}-C^{B}-\Pi^{GB}$ emerges after straightforward algebraic manipulations. This proves (P10). (P11) follows from (P4). In the following, we will show that the eigenvalues of the Jacobian matrix of the model computed in $\eta_{00}$ are real and coincident. The same holds for $\eta_{11}$. This proves (P12). Since $\eta_{00}$ and $\eta_{11}$ are the two extremes of segment $\mathcal{A}$, and invariant set by (P1), topological arguments indicate that when these two equilibria are attractors (repellors), a source (sink) in $\mathcal{A}$ is required. This proves (P13). With regard to the stability of the equilibria, let us use standard eigenvalue analysis. For system \eqref{System::Main}, we have that the Jacobian matrix is given by
\begin{equation}\label{GenericJac}
J\left(\eta^{1,G}_{t},\eta^{2,G}_{t}\right)= 
\left[
\begin{array}{cc}
J_{11}\left(\eta^{1,G}_{t},\eta^{2,G}_{t}\right) & J_{12}\left(\eta^{1,G}_{t},\eta^{2,G}_{t}\right) \\
\\
J_{21}\left(\eta^{1,G}_{t},\eta^{2,G}_{t}\right) & J_{22}\left(\eta^{1,G}_{t},\eta^{2,G}_{t}\right) 
\end{array}
\right]\text{,}
\end{equation}
where
\begin{equation}
\begin{array}{lll}
 J_{11}\left(\eta^{1,G}_{t},\eta^{2,G}_{t}\right)  &:= & \frac{\partial \eta^{1,G}_{t+1}}{\partial \eta^{1,G}_{t}} =\frac{2\eta^{1,G}_{t}}{\eta^{1,G}_{t}+\left(1-\eta^{1,G}_{t}\right)\exp\left(\beta \left(a\eta^{2,G}_{t}+b-C^{B}\right)\right)} - \frac{\left(\eta^{1,G}_{t}\right)^{2}\left(1-\exp\left(\beta \left(a\eta^{2,G}_{t}+b-C^{B}\right)\right)\right)}{\left[\eta^{1,G}_{t}+\left(1-\eta^{1,G}_{t}\right)\exp\left(\beta \left(a\eta^{2,G}_{t}+b-C^{B}\right)\right)\right]^{2}} \\
 \\
 && +  \frac{1-2\eta^{1,G}_{t}}{\eta^{1,G}_{t}+\left(1-\eta^{1,G}_{t}\right)e\exp\left(\beta\left(a\eta^{2,G}_{t}+b+C^{G}\right)\right)} -  \frac{\eta^{1,G}_{t}\left(1-\eta^{1,G}_{t}\right)\left(1-\exp\left(\beta\left(a\eta^{2,G}_{t}+b+C^{G}\right)\right)\right)}{\left[\eta^{1,G}_{t}+\left(1-\eta^{1,G}_{t}\right)\exp\left(\beta\left(a\eta^{2,G}_{t}+b+C^{G}\right)\right)\right]^{2}}\\
 \\
J_{12}\left(\eta^{1,G}_{t},\eta^{2,G}_{t}\right)  &:= & \frac{\partial \eta^{1,G}_{t+1}}{\partial \eta^{2,G}_{t}} = -\frac{\beta a \left(\eta^{1,G}_{t}\right)^{2}  \left(1-\eta^{1,G}_{t}\right)\exp\left(\beta \left(a\eta^{2,G}_{t}+b-C^{B}\right)\right)}{\left[\eta^{1,G}_{t}+\left(1-\eta^{1,G}_{t}\right)\exp\left(\beta \left(a\eta^{2,G}_{t}+b-C^{B}\right)\right)\right]^{2}} -  \frac{\beta a\eta^{1,G}_{t}\left(1-\eta^{1,G}_{t}\right)^{2}\exp\left(\beta\left(a\eta^{2,G}_{t}+b+C^{G}\right)\right)}{\left[\eta^{1,G}_{t}+\left(1-\eta^{1,G}_{t}\right)\exp\left(\beta\left(a\eta^{2,G}_{t}+b+C^{G}\right)\right)\right]^{2}}\\
\\
J_{21}\left(\eta^{1,G}_{t},\eta^{2,G}_{t}\right)  &:= & \frac{\partial \eta^{2,G}_{t+1}}{\partial \eta^{1,G}_{t}} = J_{12}\left(\eta^{2,G}_{t},\eta^{1,G}_{t}\right) \quad \text{and} \quad
J_{22}\left(\eta^{1,G}_{t},\eta^{2,G}_{t}\right) \ \ :=  \ \ \frac{\partial \eta^{2,G}_{t+1}}{\partial \eta^{2,G}_{t}} =   J_{11}\left(\eta^{2,G}_{t},\eta^{1,G}_{t}\right) \text{.}
 \end{array}
\end{equation}
Therefore,
\begin{equation}
J\left(0,0\right)= 
\left[
\begin{array}{cc}
\exp\left(-\beta\left(b+C^{G}\right)\right) & 0 \\
\\
0 & \exp\left(-\beta\left(b+C^{G}\right)\right)
\end{array}
\right]
\end{equation}
and $\lambda^{\left(0,0\right)}_{1,2}=\exp\left(-\beta\left(b+C^{G}\right)\right)>0$. Hence, $\eta_{00}=\left(0,0\right)$ is locally asymptotically stable if and only if $b+C^{G}>0$. This proves (P14). Moreover,
\begin{equation}
J\left(1,1\right)= 
\left[
\begin{array}{cc}
\exp\left(\beta\left(a+b-C^{B}\right)\right) & 0 \\
\\
0 & \exp\left(\beta\left(a+b-C^{B}\right)\right)
\end{array}
\right]
\end{equation}
and $\lambda^{\left(1,1\right)}_{1,2}=\exp\left(\beta\left(a+b-C^{B}\right)\right)>0$. Hence, $\eta_{11}=\left(1,1\right)$ is locally asymptotically stable if and only if $a+b-C^{B}<0$. This proves (P15).
Moreover,
\begin{equation}
J\left(1,0\right)= 
\left[
\begin{array}{cc}
\exp\left(\beta\left(b-C^{B}\right)\right) & 0 \\
\\
0 & \exp\left(-\beta\left(a+b+C^{G}\right)\right)
\end{array}
\right]
\end{equation}
and $\lambda^{\left(1,0\right)}_{1}=\exp\left(\beta\left(b-C^{B}\right)\right)>0$ and $\lambda^{\left(1,0\right)}_{2}=\exp\left(-\beta\left(a+b+C^{G}\right)\right)>0$. Hence, $\left(1,0\right)$ is locally asymptotically stable if and only if $b-C^{B}<0$ and $a+b+C^{G}>0$. Note that $\eta_{01}=\left(0,1\right)$ is symmetric to $\eta_{10}=\left(1,0\right)$ with respect to $\mathcal{A}$. Hence, from (P4) and the stability conditions of $\eta_{10}$, we have that $\eta_{01}$ is locally asymptotically stable if and only if $a+b+C^{G}>0$ and $b-C^{B}<0$. This proves (P16). Assuming the existence of an equilibrium in the vertex of type $\left(0,\eta^{*}\right)$, we have that
\begin{equation}
J\left(0,\eta^{*}\right)= 
\left[
\begin{array}{cc}
\exp\left(-\beta\left(a\eta^{*}+b+C^{G}\right)\right) & 0 \\
\\
J_{21}\left(0,\eta^{*}\right) & J_{22}\left(0,\eta^{*}\right)
\end{array}
\right]\text{,} 
\end{equation}
where
\begin{equation}
\begin{array}{lll}
J_{22}\left(0,\eta^{*}\right)  &:= & \frac{\left(\eta^{*}\right)^{2}+2\eta^{*}\exp\left(\beta \left(b-C^{B}\right)\right)-\left(\eta^{*}\right)^{2}\exp\left(\beta \left(b-C^{B}\right)\right)}{\left[\eta^{*}+\left(1-\eta^{*}\right)\exp\left(\beta \left(b-C^{B}\right)\right)\right]^{2}} + \frac{\left(1-\eta^{*}\right)^{2}\exp\left(\beta\left(b+C^{G}\right)\right)-\left(\eta^{*}\right)^{2}}{\left[\eta^{*}+\left(1-\eta^{*}\right)\exp\left(\beta\left(b+C^{G}\right)\right)\right]^{2}} \\
\\
& = & \frac{\left(\eta^{*}\right)^{2}\exp\left(-2\beta \left(b-C^{B}\right)\right)+2\eta^{*}\exp\left(-\beta \left(b-C^{B}\right)\right)-\left(\eta^{*}\right)^{2}\exp\left(-\beta \left(b-C^{B}\right)\right)}{\left[\eta^{*}\exp\left(-\beta \left(b-C^{B}\right)\right)+\left(1-\eta^{*}\right)\right]^{2}} + \frac{\left(1-\eta^{*}\right)^{2}\exp\left(\beta\left(b+C^{G}\right)\right)-\left(\eta^{*}\right)^{2}}{\left[\eta^{*}+\left(1-\eta^{*}\right)\exp\left(\beta\left(b+C^{G}\right)\right)\right]^{2}} \\
\\
& = & \frac{\left(\eta^{*}\right)^{2}\exp\left(-2\beta \left(b-C^{B}\right)\right)+2\eta^{*}\exp\left(-\beta \left(b-C^{B}\right)\right)-\left(\eta^{*}\right)^{2}\exp\left(-\beta \left(b-C^{B}\right)\right)+\left(1-\eta^{*}\right)^{2}\exp\left(\beta\left(b+C^{G}\right)\right)-\left(\eta^{*}\right)^{2}}{\left[\eta^{*}+\left(1-\eta^{*}\right)\exp\left(\beta\left(b+C^{G}\right)\right)\right]^{2}}\text{.}  \\
 \end{array}
\end{equation}
It follows that the eigenvalues are $\lambda_{1}^{\left(0,\eta^{*}\right)}=\exp\left(-\beta\left(a\eta^{*}+b+C^{G}\right)\right)$ and $\lambda_{2}^{\left(0,\eta^{*}\right)}=J_{22}\left(0,\eta^{*}\right)$. Note that $\lambda_{1}^{\left(0,\eta^{*}\right)}>0$ is responsible for transverse stability with respect to the edge of the box of equation $\eta^{1,G}=0$. Imposing $\lambda_{1}^{\left(0,\eta^{*}\right)}<1$, we obtain $a\eta^{*}+b+C^{G}>0$, that is, $\left(\Pi^{BG}-\Pi^{GG}\right)\eta^{*} + \left(\Pi^{BB} - \Pi^{GB}\right)\left(1-\eta^{*}\right) + C^{G}>0$. Instead, $\lambda_{2}^{\left(0,\eta^{*}\right)}$ is responsible for stability along the invariant edge $\eta^{1,G}=0$, which is the eigenvector associated with $\lambda_{2}^{\left(0,\eta^{*}\right)}$. Therefore, a necessary condition for the stability of $\left(0,\eta^{*}\right)$ is that it attracts points along this eigenvector. However, along this eigenvector, that is, on the invariant edge $\eta^{1,G}=0$ of the box, the dynamics of model \eqref{System::Main} is given by $\left(0,\eta^{2,G}_{t+1}\right) : = \left(0,H\left(\eta^{2,G}_{t}\right)\right)$, where
\begin{equation}\label{MapH}
H\left(\eta^{2,G}_{t}\right) : =  \eta^{2,G}_{t}\frac{\eta^{2,G}_{t}}{\eta^{2,G}_{t}+\left(1-\eta^{2,G}_{t}\right)\exp\left(\beta \left(b-C^{B}\right)\right)} + \left(1-\eta^{2,G}_{t}\right) \frac{\eta^{2,G}_{t}}{\eta^{2,G}_{t}+\left(1-\eta^{2,G}_{t}\right)\exp\left(\beta\left(b+C^{G}\right)\right)}
\end{equation}
on $\left[0,1\right]$. Setting $b=\Pi^{G}-\Pi^{B}$, $\eta^{2,G}_{t+1} = H\left(\eta^{2,G}_{t}\right)$ is equivalent to \eqref{MAINMAPNested1}. Therefore, $\left(0,\eta^{*}\right)$ is an equilibrium of the dynamics of model \eqref{System::Main} if and only if $\eta^{*}$ is an equilibrium of the one-dimensional model \eqref{MapH}. Moreover, the stability of $\eta^{*}$ for model \eqref{MapH} is a necessary condition for the stability of $\left(0,\eta^{*}\right)$ for the dynamics of model \eqref{System::Main}. However, by Proposition \ref{Prop::ExpReplicatorEvoluta}, we have that $\eta^{*}$ is always unstable for all $\Pi^{G},\Pi^{B}>0$, and therefore for all values of $b$. It follows that $\left(0,\eta^{*}\right)$ is either a saddle or a repellor. By (P4), the symmetric equilibrium with respect to $\mathcal{A}$, that is  $\left(\eta^{*},0\right)$, has the same stability properties of $\left(0,\eta^{*}\right)$. This proves (P17). Considering equilibrium $\left(1,\eta^{+}\right)$, we have
\begin{equation}
J\left(1,\eta^{+}\right)= 
\left[
\begin{array}{cc}
\exp\left(\beta\left(a\eta^{+}+b-C^{B}\right)\right) & 0 \\
\\
J_{21}\left(1,\eta^{+}\right) & J_{22}\left(1,\eta^{+}\right)
\end{array}
\right]\text{,} 
\end{equation}
where
\begin{equation}
\resizebox{1\hsize}{!}{$
J_{22}\left(1,\eta^{+}\right) : =    \frac{\left(\eta^{+}\right)^{2}+2\eta^{+}\left(1-\eta^{+}\right)\exp\left(\beta \left(a+b-C^{B}\right)\right)-\left(\eta^{+}\right)^{2}\exp\left(\beta \left(a+b-C^{B}\right)\right)}{\left[\eta^{+}+\left(1-\eta^{+}\right)\exp\left(\beta \left(a+b-C^{B}\right)\right)\right]^{2}}  + \frac{\left(1-\eta^{+}\right)^{2}\exp\left(\beta\left(a+b+C^{G}\right)\right)-\left(\eta^{+}\right)^{2}}{\left[\eta^{+}+\left(1-\eta^{+}\right)\exp\left(\beta\left(a+b+C^{G}\right)\right)\right]^{2}}\text{.} $}
\end{equation}
It follows that the eigenvalues are $\lambda_{1}^{\left(1,\eta^{+}\right)}=\exp\left(\beta\left(a\eta^{+}+b-C^{B}\right)\right)$ and $\lambda_{2}^{\left(1,\eta^{+}\right)}=J_{22}\left(1,\eta^{+}\right)$. Note that $\lambda_{1}^{\left(1,\eta^{+}\right)}>0$ is responsible for transverse stability with respect to the edge of the box of equation $\eta^{1,G}=1$. Imposing $\lambda_{1}^{\left(1,\eta^{+}\right)}<1$, we obtain $a\eta^{+}+b-C^{B}<0$, that is, $\left(\Pi^{BG}-\Pi^{GG}\right)\eta^{+} + \left(\Pi^{BB} - \Pi^{GB}\right)\left(1-\eta^{+}\right) - C^{B}<0$. Instead $\lambda_{2}^{\left(1,\eta^{+}\right)}$ is responsible for stability along the invariant edge $\eta^{1,G}=1$, which is the eigenvector associated with $\lambda_{2}^{\left(1,\eta^{+}\right)}$. Therefore, a necessary condition for the stability of $\left(1,\eta^{+}\right)$ is that it attracts points along this eigenvector. However, along this eigenvalue, that is, on the invariant edge $\eta^{1,G}=1$ of the box, the dynamics of model \eqref{System::Main} is given by $\left(1,\eta^{2,G}_{t+1}\right) : = \left(1,W\left(\eta^{2,G}_{t}\right)\right)$, where
\begin{equation}\label{MapW}
W\left(\eta^{2,G}_{t}\right) : =  \eta^{2,G}_{t}\frac{\eta^{2,G}_{t}}{\eta^{2,G}_{t}+\left(1-\eta^{2,G}_{t}\right)\exp\left(\beta \left(a+b-C^{B}\right)\right)} + \left(1-\eta^{2,G}_{t}\right) \frac{\eta^{2,G}_{t}}{\eta^{2,G}_{t}+\left(1-\eta^{2,G}_{t}\right)\exp\left(\beta\left(a+b+C^{G}\right)\right)}
\end{equation}
on $\left[0,1\right]$. Setting $a+b=\Pi^{G}-\Pi^{B}$, $\eta^{2,G}_{t+1} = W\left(\eta^{2,G}_{t}\right)$ is equivalent to \eqref{MAINMAPNested1}. Therefore, $\left(1,\eta^{+}\right)$ is an equilibrium of the dynamics of model \eqref{System::Main} if and only if $\eta^{+}$ is an equilibrium of the one-dimensional model \eqref{MapW}. Moreover, the stability of $\eta^{+}$ for model \eqref{MapW} is a necessary condition for the stability of $\left(1,\eta^{+}\right)$ for the dynamics of model \eqref{System::Main}. However, by Proposition \ref{Prop::ExpReplicatorEvoluta}, we have that $\eta^{+}$ is always unstable for whatever $\Pi^{G},\Pi^{B}>0$, therefore for whatever value of $b$. Hence, $\left(1,\eta^{+}\right)$ cannot attract trajectories starting at the edge $\eta^{1,G}=1$ of the box, and it is either a saddle or a repellor. It follows that $\left(1,\eta^{+}\right)$ is always unstable, and either a saddle or a repellor. By (P4), the symmetric equilibrium with respect to $\mathcal{A}$, that is  $\left(\eta^{+},1\right)$, has the same stability properties of $\left(1,\eta^{+}\right)$. This proves (P18). Consider an equilibrium in the invariant set $\mathcal{A}$. By definition of $\mathcal{A}$ in (P1), this equilibrium is given by $\left(\bar{\eta},\bar{\eta}\right)$, with $\bar{\eta}\in\left(0,1\right)$. Hence, the Jacobian matrix in \eqref{GenericJac} becomes
\begin{equation}
J\left(\bar{\eta},\bar{\eta}\right)= 
\left[
\begin{array}{cc}
A & B \\
\\
B & A 
\end{array}
\right]\text{,}
\end{equation}
where $A = J_{11}\left(\bar{\eta},\bar{\eta}\right) = J_{22}\left(\bar{\eta},\bar{\eta}\right)$ and $B = J_{12}\left(\bar{\eta},\bar{\eta}\right) = J_{21}\left(\bar{\eta},\bar{\eta}\right)$. It follows that the eigenvalues are always real and given by $\lambda^{\left(\bar{\eta},\bar{\eta}\right)}_{1,2}=A\pm B$, with \eqref{Eigenvaluesetaeta}, which follows from straightforward algebra, and $v_{1,2}=\left[\pm 1, 1\right]$ are a couple of two related eigenvectors, respectively. This proves (P19). For $\left(\bar{\eta},\bar{\eta}\right)\in\left(0,1\right)^{2}$, note that
\begin{equation}
A = \frac{\left(\bar{\eta}\right)^{2}+\bar{\eta}\left(2-\bar{\eta}\right)\exp\left(\beta \left(a\bar{\eta}+b-C^{B}\right)\right)}{\left[\bar{\eta}+\left(1-\bar{\eta}\right)\exp\left(\beta \left(a\bar{\eta}+b-C^{B}\right)\right)\right]^{2}}  +    \frac{\left(1-\bar{\eta}\right)^{2}\exp\left(\beta\left(a\bar{\eta}+b+C^{G}\right)\right)  
  -\left(\bar{\eta}\right)^{2}}{\left[\bar{\eta}+\left(1-\bar{\eta}\right)\exp\left(\beta\left(a\bar{\eta}+b+C^{G}\right)\right)\right]^{2}}  >0 \text{,} 
\end{equation}
since $\left[\bar{\eta}+\left(1-\bar{\eta}\right)\exp\left(\beta\left(a\bar{\eta}+b+C^{G}\right)\right)\right]^{2}>\left[\bar{\eta}+\left(1-\bar{\eta}\right)\exp\left(\beta \left(a\bar{\eta}+b-C^{B}\right)\right)\right]^{2}$. Moreover, assume $a>0$ and note that
\begin{equation}
B = -\frac{\beta a \left(\bar{\eta}\right)^{2}  \left(1-\bar{\eta}\right)\exp\left(\beta \left(a\bar{\eta}+b-C^{B}\right)\right)}{\left[\bar{\eta}+\left(1-\bar{\eta}\right)\exp\left(\beta \left(a\bar{\eta}+b-C^{B}\right)\right)\right]^{2}} -  \frac{\beta a\bar{\eta}\left(1-\bar{\eta}\right)^{2}\exp\left(\beta\left(a\bar{\eta}+b+C^{G}\right)\right)}{\left[\bar{\eta}+\left(1-\bar{\eta}\right)\exp\left(\beta\left(a\bar{\eta}+b+C^{G}\right)\right)\right]^{2}}<0 \text{.}    
\end{equation}
Hence, $\lambda^{\left(\bar{\eta},\bar{\eta}\right)}_{2}>0$ and $\lambda^{\left(\bar{\eta},\bar{\eta}\right)}_{2}>\left|\lambda^{\left(\bar{\eta},\bar{\eta}\right)}_{1}\right|$. By standard bifurcation analysis, see, e.g., \citet{GuckenheimerHolmes1983}, it follows that an inner equilibrium $\left(\bar{\eta},\bar{\eta}\right)$ can undergo only fold/transcritical bifurcations along the direction spanned by eigenvector $v_{2}$, that is, along the manifold transverse to the invariant set $\mathcal{A}$. Along the invariant set $\mathcal{A}$, the direction spanned by eigenvector $v_{1}$, an equilibrium $\left(\bar{\eta},\bar{\eta}\right)$ can undergo either a bifurcation of eigenvalue $1$ (when $\lambda^{\left(\bar{\eta},\bar{\eta}\right)}_{1}=1$) or a bifurcation of eigenvalue $-1$ (when $\lambda^{\left(\bar{\eta},\bar{\eta}\right)}_{1}=-1$). Note that $\left(\bar{\eta},\bar{\eta}\right)$ is unstable (a saddle or a replellor) at least immediately before and after these bifurcations. Indeed, $\lambda^{\left(\bar{\eta},\bar{\eta}\right)}_{1}$ equal to either $1$ or $-1$ implies $\lambda^{\left(\bar{\eta},\bar{\eta}\right)}_{2}>1$, since $\lambda^{\left(\bar{\eta},\bar{\eta}\right)}_{2}>\left|\lambda^{\left(\bar{\eta},\bar{\eta}\right)}_{1}\right|$. For $a<0$ and $B>0$, similar but opposite considerations apply. This proves (P20).
\end{proof}

\medskip

\begin{proof}[Proof of Proposition \ref{Prop:FullModelSce}]
As shown in Lemma \ref{Lemma:MainProperties}, the stability of the vertex equilibria depends on the sign of four different expressions: 1) $\Pi^{BB} - \Pi^{GB} + C^{G}$, 2) $\Pi^{GB}-\Pi^{BB} + C^{B}$, 3) $\Pi^{GG} - \Pi^{BG}+C^{B}$, 4) $\Pi^{BG}-\Pi^{GG}+C^{G}$. Therefore, there are $2^{4}=16$, possible scenarios. However,
$\Pi^{BB} < \Pi^{GB} - C^{G}$ and $\Pi^{GB}<\Pi^{BB} -C^{B}$ lead to $\Pi^{GB}<\Pi^{BB} -C^{B}<\Pi^{BB} < \Pi^{GB} - C^{G}$, which is not possible since $C^{G}\geq0$ by assumption. Moreover, $\Pi^{GG} < \Pi^{BG}-C^{B}$ and $\Pi^{BG}<\Pi^{GG}-C^{G}$ lead to $\Pi^{GG} < \Pi^{BG}-C^{B}<\Pi^{BG}<\Pi^{GG}-C^{G}$, which is not possible since $C^{G}\geq0$ by assumption. By straightforward consideration, we then have that all possible scenarios reduce to $9$. The other results follow from Lemma \ref{Lemma:MainProperties}.
\end{proof}

\medskip

\begin{proof}[Proof of Corollary \ref{Corr2}]
Consider the four conditions that characterize the nine scenarios of Proposition \ref{Prop:FullModelSce}: $\Pi^{BB} \lessgtr \Pi^{GB} - C^{G}$, $\Pi^{GG} \lessgtr \Pi^{BG}-C^{B}$, $\Pi^{BG} \lessgtr  \Pi^{GG}-C^{G}$ and $\Pi^{GB} \lessgtr \Pi^{BB} -C^{B}$. If $\Pi^{GG}=\Pi^{GB} = \Pi^{G}$, and $\Pi^{BG}=\Pi^{BB} = \Pi^{B}$, these conditions reduce to: $\Pi^{B} \lessgtr \Pi^{G} - C^{G}$ and $\Pi^{G} \lessgtr \Pi^{B} -C^{B}$. These conditions are consistent only with Scenarios 1, 5, and 9. This proves (1). If $\Pi^{GG} = \Pi^{GB} = \Pi^{G}$, the four conditions that characterize the nine scenarios of Proposition \ref{Prop:FullModelSce} become: $\Pi^{BB} \lessgtr \Pi^{G} - C^{G}$, $\Pi^{G} \lessgtr \Pi^{BG}-C^{B}$, $\Pi^{BG} \lessgtr  \Pi^{G}-C^{G}$, and $\Pi^{G} \lessgtr \Pi^{BB} -C^{B}$. Imposing $\Pi^{BB}<\Pi^{BG}$, we have: (A): $\Pi^{G} < \Pi^{BB}-C^{B}$, then $\Pi^{G} < \Pi^{BG} -C^{B}$; (B) $\Pi^{G} > \Pi^{BG}-C^{B}$, then $\Pi^{G} > \Pi^{BB} -C^{B}$; (C) $\Pi^{BB} >\Pi^{G} - C^{G}$, then $\Pi^{BG} > \Pi^{G}-C^{G}$; (D) $\Pi^{BG} < \Pi^{G}-C^{G}$, then $\Pi^{BB} <\Pi^{G} - C^{G}$. Therefore, Scenarios 2, 3, and 6 are not possible. This proves (2). If $\Pi^{BB} = \Pi^{BG} = \Pi^{B}$, the four conditions that characterize the nine scenarios of Proposition \ref{Prop:FullModelSce} become: $\Pi^{B} \lessgtr \Pi^{GB} - C^{G}$, $\Pi^{GG} \lessgtr \Pi^{B}-C^{B}$, $\Pi^{B} \lessgtr  \Pi^{GG}-C^{G}$, and $\Pi^{GB} \lessgtr \Pi^{B} -C^{B}$. Imposing $\Pi^{GG}>\Pi^{GB}$, we have: (A): $\Pi^{B} < \Pi^{GB} - C^{G}$, then $\Pi^{B} < \Pi^{GG} - C^{G}$; (B) $\Pi^{B} > \Pi^{GG} - C^{G}$, then $\Pi^{B} > \Pi^{GB} - C^{G}$; (C) $\Pi^{GG} < \Pi^{B}-C^{B}$, then $\Pi^{GB} < \Pi^{B} -C^{B}$; (D) $\Pi^{GB} > \Pi^{B} -C^{B}$, then $\Pi^{GG} > \Pi^{B} -C^{B}$. Therefore, Scenarios 4, 7, and 8 are not possible. This proves (3). Imposing $\Pi^{BG}>\Pi^{BB}$ and $\Pi^{GG}>\Pi^{GB}$, neither of the conditions that characterize one the nine scenarios can be excluded. This proves (4).
\end{proof}

\medskip

\begin{proof}[Proof of Corollary \ref{Corr3}]
Note that $\Pi^{GG}>\Pi^{GB} > \Pi^{BG}> \Pi^{BB}$ and $C^{B}>0$ imply $\Pi^{GG}> \Pi^{BG} - C^{B}$ and $\Pi^{GB}> \Pi^{BB} - C^{B}$. Hence, only Scenarios 1, 6,  and 9 are possible. Further assuming $C^{G}=0$ (subsidizing the cost of adapting green technology), implies $\Pi^{BB} < \Pi^{GB} - C^{G}$ and $\Pi^{BG} < \Pi^{GG} - C^{G}$. Hence, only Scenario 9 is possible. Assuming $C^{G}>\max\left\{\Pi^{GB}-\Pi^{BB};\Pi^{GG}-\Pi^{BG}\right\}$ implies $\Pi^{BB} > \Pi^{GB} - C^{G}$ and $\Pi^{BG} > \Pi^{GG} - C^{G}$. Hence, only Scenario 1 is possible. This proves (1). Note that $\Pi^{BG}> \Pi^{BB}>\Pi^{GG}>\Pi^{GB}$ and $C^{G}>0$ imply $\Pi^{BB}> \Pi^{GB} - C^{G}$ and $\Pi^{BG}> \Pi^{GG} - C^{G}$. Hence, only Scenarios 1, 2, 4, and 5 are possible. Further assuming $C^{B}=0$ implies $\Pi^{GG} < \Pi^{BG} - C^{B}$ and $\Pi^{GB} < \Pi^{BB} - C^{B}$. Hence, only Scenario 5 is possible. Assuming $C^{B}>\max\left\{\Pi^{BG}-\Pi^{GG};\Pi^{BB}-\Pi^{GB}\right\}$ implies $\Pi^{GG} > \Pi^{BG} - C^{B}$ and $\Pi^{GB} > \Pi^{BB} - C^{B}$. Hence, only Scenario 1 is possible. This proves (2).
\end{proof}

\medskip



\end{document}